\documentclass[11pt,letterpaper]{article}
\usepackage{amsmath,amsfonts,amsthm,amssymb}
\usepackage{setspace}
\usepackage{fancyhdr}
\usepackage{lastpage}
\usepackage{extramarks}
\usepackage{xspace}
\usepackage{chngpage}

\usepackage[ruled,linesnumbered]{algorithm2e} 

\usepackage{soul,color}
\usepackage{graphicx,float,wrapfig}
\usepackage[font=small,labelfont=bf]{caption}
\usepackage[margin=1in]{geometry}
\usepackage{enumitem}
\usepackage{thmtools,thm-restate}
\usepackage{bbm}

\setlist{nosep}

\allowdisplaybreaks

\usepackage{mdframed}
\newtheorem*{mdresult}{Result}

\usepackage[dvipsnames]{xcolor}
\usepackage[linktocpage=true,breaklinks,colorlinks,citecolor=blue,linkcolor=BrickRed]{hyperref}

\usepackage{wrapfig}
  
\title{\textmd{\bf Learning Anonymous Pricing  for \\  Online Resource Allocation
  }}

\date{\today}
\author{Yifeng Teng\footnote{Google Research. Email: yifengt@google.com.}\and  Yifan Wang\footnote{School of Computer Science, Georgia Institute of Technology, Atlanta, GA, USA. Email: ywang3782@gatech.edu. Supported in part by NSF awards CCF-2327010 and CCF-2440113.}}

\usepackage[usenames,dvipsnames]{xcolor}
\usepackage[linktocpage=true,breaklinks,colorlinks,citecolor=blue,linkcolor=BrickRed]{hyperref}
\usepackage{cleveref}

\DeclareUnicodeCharacter{2217}{*}

\newtheorem{Theorem}{Theorem}[section]
\newtheorem{Lemma}[Theorem]{Lemma}

\newtheorem{Claim}[Theorem]{Claim}

\newcommand{\parta}{(\text{\uppercase\expandafter{\romannumeral1}})}
\newcommand{\partb}{(\text{\uppercase\expandafter{\romannumeral2}})}
\newcommand{\partc}{(\text{\uppercase\expandafter{\romannumeral3}})}
\newcommand{\be}{\mathtt{BE}}
\newcommand{\bd}{\mathtt{Good}}
\newcommand{\pbe}{P_\mathtt{BE}}
\newcommand{\opt}{\mathsf{Opt}}
\newcommand{\var}{\mathsf{Var}}

\newcommand{\wel}{\mathsf{Wel}}

\newcommand{\alg}{\mathsf{Alg}}
\newcommand{\algt}{\mathsf{Alg}^{(T)}}

\newcommand{\F}{\mathcal{F}}
\newcommand{\E}{\mathbb{E}}

\newcommand{\one}{\mathbf{1}}%
\newcommand{\pr}{\mathbf{Pr}} 

\newcommand{\ignore}[1]{{}}
\newcommand{\R}{\mathbb{R}}

\newcommand{\D}{\mathcal{D}}
\newcommand{\bD}{\boldsymbol{\mathcal{D}}}
\newcommand{\bgamma}{\boldsymbol{\gamma}}

\newcommand{\poly}{\mathsf{poly}}

\newcommand{\eps}{\epsilon}

\newcommand{\IGNORE}[1]{}

\newcommand{\price}{\lambda}
\newcommand{\initprice}{\lambda^{\mathtt{INIT}}}
\newcommand{\esopt}{\widehat{\mathsf{Opt}}}
\newcommand{\val}{v}
\newcommand{\bprice}{\pmb{\lambda}}
\newcommand{\bpricet}{{\pmb{\lambda}}^{(t)}}

\newcommand{\B}{\mathbf{B}}
\newcommand{\calD}{\mathcal{D}}

\newcommand{\alloc}{\boldsymbol{a}}
\newcommand{\cons}{\boldsymbol{c}}
\newcommand{\consbase}{c}
\newcommand{\realallocbase}{{d}^{(t)}}

\newcommand{\realallocit}{\boldsymbol{d}^{(t)}_{i}}
\newcommand{\realalloct}{\boldsymbol{d}^{(t)}}

\newcommand{\hist}{\mathcal{H}}

\begin{document}

\maketitle \thispagestyle{empty}

\begin{abstract}






We study the online resource allocation problem, where a seller sequentially receives independent requests for $m$ types of resources with limited supplies from $n$ heterogeneous agents arriving in an \emph{unknown} order. Each request from an agent can be fulfilled in different ways, with resource consumption in $[0,1]^m$, and generates different values for the agent. The objective of the seller is to maximize the social welfare, which is the sum of the values obtained from each agent.

Recently, Ghuge, Singla, and Wang \cite{GSW-STOC25} studied the learnability of the online resource allocation problem with heterogeneous agents and proposed a learnable pricing algorithm using only a single sample. However, their core algorithm is a \emph{dynamic} pricing algorithm, which may introduce fairness concerns, as different agents face different prices. Furthermore, the algorithm crucially needs to know the arrival order of the agents in advance. To address these issues, in this paper, we study the learnability of \emph{anonymous} pricing algorithms for online resource allocation using samples and queries to agents’ value distributions.
First, we show that a polynomial number of samples suffices to learn the classic dual pricing algorithm. Second, we show that a polynomial number of pricing queries suffices to learn a near-optimal anonymous pricing algorithm, in which the item pricing vector faced by each agent is drawn from the same predetermined distribution.






\end{abstract}

\newpage




\section{Introduction}


Consider an online resource allocation setting where a seller offers $m$ shared resources for sale, and $n$ buyers with requests having different preferences over combinations of resources arrive over time. We want to design an online algorithm for allocating the resources to buyers to satisfy their request in order to maximize \textit{social welfare}, which is the total value obtained by all buyers, while satisfying the predefined budget constraints. When each agent arrives, the seller makes an irrevocable allocation decision and does not wait for future arrivals. 

We study a general stochastic setting where the request type of each buyer is drawn independently from unknown and possibly \textit{non-identical} distributions. Recently, \cite{GSW-STOC25} proved the first bounded sample complexity result for such an online resource allocation problem with heterogeneous agents, that under a large-budget constraint (where each type of resource has $\widetilde \Omega (\poly(\eps))$ supply), it is possible to design a $(1-\eps)$-approximation algorithm given just a single sample from each of the $n$ request distributions\footnote{We remark that it is possible to obtain a $(1-\eps)$-approximation for online resource allocation problem only when the supply of each resource is at least $\Omega(\eps^{-2})$, even if there is only single resource; see, e.g., \cite{kleinberg2005multiple}.}. The proposed algorithm is a \textit{dynamic pricing} algorithm, that decides the prices of the items according to the previous arrivals and the current arriving agent. In particular, the algorithm proposed is based on a pricing scheme, where the price of a resource adjusts exponentially based on the difference between its actual consumption and its expected consumption up to the current time. While this yields a strong theoretical guarantee with minimal samples needed, the resulting mechanism introduces the following two limitations.
\begin{itemize}
    \item \textit{Lack of Anonymity:} Two identical agents requesting the same resources with the same value may face different prices simply because they arrived at different positions in the queue. This form of discrimination is often undesirable.
    \item \textit{Sensitivity to Ordering:} The learning algorithm proposed by \cite{GSW-STOC25} requires knowledge of the arrival sequence to compute the expected consumption at each step. It is not robust to an unknown or adversarial arrival order.
\end{itemize}

Given the limitations mentioned above, we want to understand whether a simpler mechanism is learnable. In particular, we ask the following question, arising from the above limitations.

\begin{quote}
\emph{Given \textbf{sample or query} access to the agents' value distributions, is it possible to learn an \textbf{anonymous pricing} mechanism for agents arriving in \textbf{unknown} order with a small number of samples that obtains near-optimal welfare?}
\end{quote}

\subsection{Model}

Before we introduce our main result, we start with introducing the model we study. We will define the online resource allocation problem that we study, different types of item pricing mechanisms, and what we mean by learning through queries.

\subsubsection{Online Resource Allocation}

We study an online resource allocation problem, where the seller receives a sequence of $n$ requests from buyers arriving in arbitrary order, for $m$ kinds of limited resources with budget $\B=(B_1,B_2,\cdots,B_m)\in \R^m_{\geq 0}$ for all resources. The request from each buyer $i \in [n]$ is denoted by $\gamma_i = (\val_i, \alloc_i, \Theta_i)$, where $\Theta_i$ denotes the decision set, and choosing any decision $\theta_i\in\Theta_i$ will generate value $\val_i(\theta_i)\in\R_{\geq 0}$ for buyer $i$, while consuming at most a unit of each resource $\alloc_i(\theta_i)=(a_{i,1}(\theta_i),a_{i,2}(\theta_i),\cdots,a_{i,m}(\theta_i))\in[0,1]^m$. Without loss of generality assume $B_j\leq n$ for each resource $j$. We use $|\Theta_{\max}|$ to denote the maximum size of the decision set, and always assume that there is a null decision $\phi\in\Theta_i$ for each buyer, such that $\val_i(\phi)=0$ and $\alloc_i(\phi)=\mathbf{0}$. The objective of the seller is to maximize social welfare $\sum_{i=1}^{n} \val_i(\theta_i)$, while not exceeding the budget of any resource.


We study the stochastic setting with underlying request distributions $\bD = (\D_1, \cdots, \D_n)$, i.e.,  the request $\gamma_i$ of each buyer is drawn independently from a type distribution $\calD_i$ over a type space $\Gamma_i$. Each type $\gamma_i = (\val_i, \alloc_i, \Theta_i)\sim\calD_i$ is realized for each agent $i \in [n]$. The benchmark of the optimization problem is the \emph{hindsight fractional optimal solution} $\opt$, which is the optimal value of the configuration linear program defined below. Let $x_{i,\gamma_i, \theta}$ denote the fractional allocation to agent $i$ with type $\gamma_i$ and decision $\theta \in \Theta_i$. Then we can write the configuration LP as follows:
\begin{align}
    \opt\quad:=\quad\text{maximize } \quad \quad  &  \textstyle  \quad\sum_{i \in [n]}  \E_{\gamma_i \sim \calD_i} \left[ \sum_{\theta \in \Theta_i} \val_i(\theta) \cdot x_{i,\gamma_i, \theta} \right],\notag \\
    \text{s.t.}\qquad\qquad & \textstyle \quad   \sum_{i \in [n]} \E_{\gamma_i \sim \calD_i} \left[ \sum_{\theta \in \Theta_i} \alloc_i(\theta) \cdot x_{i,\gamma_i, \theta} \right]  \leq \B, \tag{$\text{LP}_{\textsc{Relax}}$} \label{program:ora-ub}  \\
    \forall i \in [n], \gamma_i\in\Gamma_i, &  \textstyle \quad \sum_{\theta \in \Theta_{i}} x_{i,\gamma_i, \theta} \leq 1, \notag \\
    \forall i \in [n],  \theta \in \Theta_{i},&  \quad  0 \leq x_{i,\gamma_i, \theta} \leq 1. \notag
\end{align}

We use notation \ref{program:ora-ub}$(\bD; \B)$ to represent the linear program \ref{program:ora-ub} with underlying request distributions $\bD$ and budget $\B$.
Then, given an online resource allocation instance with underlying distributions $\bD$ and budget $\B$, the fractional solution of \ref{program:ora-ub}$(\bD; \B)$ is an upper bound of the hindsight optimal solution (i.e. when the seller knows the realization and arrival order of the agents in advance).

Next, we characterize the dual of the primal program \ref{program:ora-ub}. Define the dual variables as follows. Let $\bprice \in \mathbb{R}^m_{\geq 0}$ be the dual variable associated with the resource budget constraint. For each $i \in [n]$ and each type $\gamma_i \in \Gamma_i$, let $u_{i, \gamma_i} \geq 0$ be the dual variable associated with the demand constraint $\sum_{\theta \in \Theta_i} x_{i, \gamma_i, \theta} \leq 1$. The dual problem can be written as follows:

\begin{align}
    \opt_{\textsc{dual}} \quad:=\quad \text{minimize } \quad \quad & \langle\bprice,\mathbf{B}\rangle + \sum_{i \in [n]} \mathbb{E}_{\gamma_i \sim \mathcal{D}_i} [u_{i, \gamma_i}] \notag \\
    \text{s.t.} \qquad \qquad  & \bprice \geq 0, \notag \\
    \forall i \in [n], \gamma_i \in \Gamma_i, \theta \in \Theta_i,\quad & u_{i, \gamma_i} + \langle\bprice,\alloc_i(\theta)\rangle \geq \val_i(\theta), \tag{$\text{LP}_{\textsc{dual}}$} \label{program:dual} \\
    \forall i \in [n], \gamma_i \in \Gamma_i,\quad & u_{i, \gamma_i} \geq 0. \notag
\end{align}

The dual variables $\bprice$ in the optimal dual solution correspond to the item prices such that under a specific tie-breaking rule, setting static pricing $\bprice$ to all agents obtains the optimal welfare. Similarly, we use notation \ref{program:dual}$(\bD; \B)$ to represent the linear program \ref{program:dual} with underlying request distributions $\bD$ and budget $\B$.




\subsubsection{Pricing Mechanisms}
In this paper, when designing the allocation mechanisms for each agent, we consider simple \emph{posted pricing mechanisms}. When each agent $i$ with type $\gamma_i=(v_i,\alloc_i,\Theta_i)$ arrives, she faces a price vector $\bprice_i\in\R_{\geq 0}^m$ with a non-negative price $\lambda_{i,j}$ for each resource $j\in[m]$. The agent then selects the utility-maximizing decision $\theta_i\in\Theta_i$, i.e.
$$\theta_i=\mathop{\arg\max}_{\theta\in\Theta_i} \left(v_i(\theta)-\langle \alloc_i(\theta),\bprice_i\rangle\right),$$
which is the decision that maximizes the difference between the agent's value and the price for the resource consumption.\footnote{We require that the agent respects the inventory constraints, that the agent makes the utility-maximizing decision with resource consumption capped by the current remaining units for each resource. While this may result in overload cascades in the online resource allocation problems \cite{chawla2017stability}, in the analysis of our algorithms we will treat the overload of any resource as a bad event in the online process, and exclude the overall welfare contribution whenever a bad event happens.} In this paper, we discuss the following different types of pricing mechanisms:
\begin{itemize}
\item \textbf{Adaptive (Dynamic) Pricing.} When each agent $i$ arrives, the seller sets a price vector $\bprice_i\in\R_{\geq 0}^m$,  depending on the history of previous allocations and possibly the identity of the arriving agent. Then the arriving agent makes a utility-maximizing decision $\theta_i$.
\item \textbf{Anonymous Pricing.} When any agent $i$ arrives at time $t$, the seller sets a pricing $\bprice_t$ that does not depend on the allocation history, the buyer's identity, and the current time. Then the arriving agent makes a utility-maximizing decision $\theta_i$.
\item \textbf{Static Pricing.} Before any agent arrives, the seller sets a fixed pricing $\bprice\in\R_{\geq 0}^m$. Any arriving agent $i$ faces the same price vector $\bprice$ for the resources, and makes a utility-maximizing decision $\theta_i$. \footnote{It can be a little confusing that \emph{Static Pricing} and \emph{Anonymous Pricing} are defined differently. In fact, this is due to the fact that \emph{Static Pricing} has a canonical meaning in the literature, that refers to a fixed deterministic pricing mechanism that does not change over time. It is not the same as the meaning of ``static'' in the online learning literature, which corresponds to algorithms that do not change over time. Therefore, in this paper we use \emph{Anonymous Pricing} to denote the pricing algorithms that are possibly randomized, but are oblivious to the arrival order, allocation history, and the identity of the current agent.}
\end{itemize}

In our main results, we allow \textit{adversarial} tie breaking, so that when the utility of two decisions is the same for an arriving agent, the agent can choose an arbitrary decision. We may use consistent tie breaking to derive some intermediate results. 

\subsubsection{Learning from Samples / Queries}
We study the setting where the type distribution $\calD_i$ of each buyer $i$ is not fully available to the seller, but is available via two types of \emph{query access}. The first type of access model is the classic \textit{\textbf{sample access}} model, where the learner is allowed to query each type distribution $\calD_i$ of agent $i$, and can obtain a fresh type sample $\gamma_i\sim\calD_i$ in each query. The second type of access model is the \textit{\textbf{pricing query access}} model (and we may slightly abuse the notation and call it the \textbf{\textit{query access}} model when the context is clear), where the learner is allowed to query each type distribution $\calD_i$ with a price vector $\bprice$. Then a type sample $\gamma_i\sim\calD_i$ is drawn from the type distribution of an agent, but the learner is only able to observe any utility-maximizing decision of the agent with type $\gamma_i$. In other words, the output of the query comes from a demand oracle $\mathcal{O}$ that, given any $\bprice$ and $\gamma_i$, $\mathcal{O}$ returns an arbitrary decision $\theta_i\in \Theta_i$ in the utility maximizing set $\mathop{\arg\max}_{\theta\in\Theta_i} \left(v_i(\theta)-\langle \alloc_i(\theta),\bprice_i\rangle\right)$.\footnote{As noted by \cite{GSW-STOC25} this is enough to solve the configuration LP efficiently.} 
The goal of the learner is to learn a pricing strategy, such that for any arrival order of the agent, the expected welfare is at least $(1-\eps)\opt$ under the budget constraint $\B$.



\subsection{Main Results}

In this paper, we give an affirmative answer to the question we asked in two slightly different settings. 
Our first result shows that when each agent has a \textit{generic} valuation distribution \footnote{The \textit{generic} property is the condition that, for any item pricing, the probability that at least two decisions of an agent are both utility-maximizing is 0. This kind of atomless property is usually needed to deal with the tie breaking (Eg. in \cite{FOCS24-toapper}). We will discuss more in Section~\ref{subsec:discussion-static} why this is needed.}, it is possible to learn a \textit{static item pricing} through polynomial number of \textit{samples}. In particular, it is close to the dual prices computed from the configuration LP of the optimization problem. 

\begin{Theorem}[Informal \Cref{thm:sample-main}]\label{thm:sample-main-intro}
Given $\poly(m,\eps^{-1}, \log n, \log |\Theta|_{\max})$ samples from the agents' value distributions, where $|\Theta|_{\max}$ represents the maximum size of the action set, we can compute a \textbf{static} item pricing mechanism that obtains a $(1-\eps)$ fraction of optimal welfare in hindsight for the online resource allocation problem with non-identical and generic distributions, provided that each resource has $\Omega(\poly( \log(mn\eps^{-1})) \cdot \epsilon^{-2})$
budget. 
\end{Theorem}

We remark that our sample complexity depends on the maximum size $|\Theta|_{\max}$ of each request’s action set. While this dependence may prevent a direct application of our result to settings with continuous action sets, a logarithmic dependence on $|\Theta|_{\max}$ still suffices to yield a polynomial sample complexity bound for a broad class of online resource allocation problems. In particular, our result applies even when the action set is exponentially large, such as in online combinatorial auctions, where $|\Theta|_{\max} \le 2^m$. We present our proofs of the sample complexity bound in \Cref{sec:sample}.

Our second result shows that if we only have \textit{query} access to the value distributions, we can still obtain an \textit{anonymous} pricing mechanism through a polynomial number of queries, even for request distributions beyond \textit{generic}.

\begin{Theorem}[Informal \Cref{thm:main-formal}]\label{thm:query-main-intro}
Given $T = n$ pricing queries from each of the $n$ request distributions, we can compute an \textbf{anonymous} item pricing mechanism that obtains a $(1-\eps)$ fraction of optimal welfare in hindsight for the online resource allocation problem with non-identical distributions, provided that each resource has $\Omega(\poly( \log(mn\eps^{-1})) \cdot \epsilon^{-2})$
budget. 
\end{Theorem}

The detailed algorithm underlying our query complexity result is presented in \Cref{sec:query-complexity}. 
The anonymous pricing mechanism we construct takes the following form. 
We first learn, from samples, a set of $T$ price vectors $\Lambda = \{\bprice^{(1)}, \bprice^{(2)}, \ldots, \bprice^{(T)}\} \subseteq \mathbb{R}^m$. 
Upon the arrival of each agent, the seller selects a price vector from $\Lambda$ uniformly at random and posts it to the agent.

This mechanism is anonymous and non-adaptive: it does not depend on the identity of the arriving agent or on the current state of the remaining resources. 
Moreover, despite having access only to limited query feedback, our algorithm does not require solving the computationally expensive configuration LP~(\ref{program:ora-ub}). 
To the best of our knowledge, this is the first algorithm for online resource allocation with non-identical requests that avoids solving the configuration LP altogether.

\subsection{Technical Overview}
\label{sec:technical-overview}

\vspace{0.5em} \noindent \textbf{Ideas for sample complexity.} Our proof of a polynomial sample complexity consists of the following three ingredients:
\begin{itemize}
    \item \textbf{Converting the learning objective.} 
    Instead of aiming to learn a price vector that approximates the dual prices, i.e., the optimal solution to~\ref{program:dual}, we focus on learning a price vector whose expected consumption is close to, but does not exceed, the budget constraint~$\B$. 
    By applying weak duality, we show that ensuring the expected consumption is close to~$\B$ is sufficient to guarantee an expected welfare close to~$\opt$.

    \item \textbf{Uniform convergence via pseudo-dimension.} 
    We bound the sample complexity by determining the number of samples required to construct an empirical distribution such that the expected consumption under the empirical distribution is close to that under the true distribution. 
    To this end, we follow a similar approach of~\cite{FOCS24-toapper} and establish a uniform convergence result for the expected consumption over all price vectors by bounding the pseudo-dimension of the entire allocation process.
    
    \item \textbf{Improving sample complexity via truncating the consumption.} 
    Observe that the expected consumption takes values in the range $[0, n]$. 
    A uniform convergence guarantee over this range would lead to a sample complexity that depends polynomially on~$n$. 
    However, our interest is restricted to price vectors whose expected consumption is close to~$B_j$, which may be much smaller than~$n$. 
    To improve the sample complexity, we truncate the consumption at $O(B_j)$ and show that establishing uniform convergence for the truncated consumption suffices. 
    This refinement yields an improved sample complexity bound that no longer depends polynomially on~$n$.
\end{itemize}

\vspace{0.5em} \noindent \textbf{Ideas for pricing query complexity.} The main algorithm for pricing query feedback is a learning algorithm that runs for $T$ rounds. 
In each round $t$, we present an identical price vector $\bprice^{(t)}$ to $n$ requests independently drawn from distributions $\D_1, \ldots, \D_n$. Based on the query feedback, i.e., the consumption of each request, we update the price vector to $\bprice^{(t+1)}$ for round $t+1$.

It remains to specify the rule for updating the price vector. 
Our main intuition is to view the above learning algorithm as a large online resource allocation problem with identical distributions: there are $T$ identical large requests, $m$ resources, and each resource $j$ has a total budget of $T \cdot B_j$. 
The $t$-th large request corresponds to the $n$ small requests independently sampled in round $t$, and hence all large requests are identically distributed.

To update prices for this online resource allocation problem with identical distributions, we apply the \emph{exponential pricing} algorithm introduced in \cite{GSW-STOC25}. 
Exponential pricing is a dynamic pricing algorithm that aims to match the optimal expected consumption, with the key idea being to adjust prices \emph{exponentially} according to the discrepancy between the current consumption and the optimal expected consumption. Since all large requests are identical, the optimal expected consumption for each request is $\B$. 
Applying this optimal expected consumption, scaled down by a factor of $1 - \epsilon$, yields our price update rule.

Since exponential pricing achieves a $1 - O(\epsilon)$ competitive ratio when the budget is at least $\widetilde{\Omega}(\epsilon^{-2})$ times the maximum consumption of a single request (see \cite{GSW-STOC25}), our learning algorithm achieves a total welfare of $(1 - \epsilon)\cdot T \cdot \opt$ whenever
\[
T \cdot B_j \ge \widetilde{\Omega}(\epsilon^{-2}) \cdot n,
\]
noting that each large request consumes at most $n$ units of any resource. 
This condition in particular implies that $T \ge n$.

The final step of the proof is to convert the price vectors learned over the $T$ rounds into an anonymous pricing mechanism. 
Specifically, we show that independently selecting a price vector uniformly at random from the set $\{\bprice^{(1)}, \ldots, \bprice^{(T)}\}$ and presenting it to incoming requests still achieves a $1 - O(\epsilon)$ competitive ratio. Establishing this result constitutes the main technical challenge we face: \cite{GSW-STOC25} establishes strong performance guarantees for the exponential pricing algorithm in a dynamic online resource allocation setting. 
However, in our setting, the learned price vector set $\{\bprice^{(1)}, \ldots, \bprice^{(T)}\}$ is subsequently evaluated on a fresh batch of requests. This requires an additional guarantee showing that prices learned via exponential pricing continue to perform well when the online resource allocation instance is effectively re-run on new data.

We establish such a guarantee by modeling the exponential pricing process as a martingale. 
Using a martingale-based concentration bound together with careful analysis to bypass some correlation issues we introduced during the analysis, we show that the performance in the re-running phase and the learning phase both cannot deviate significantly from its expectation. 
As a result, the welfare achieved in the re-running phase is comparable to that attained during the learning phase.

\subsection{Discussion on Static Pricing Assumptions}
\label{subsec:discussion-static}

If we compare Theorem~\ref{thm:sample-main-intro} and Theorem~\ref{thm:query-main-intro}, we find that to get the static pricing result, we need a stronger access model (sample vs. query) and the assumption that the distribution is generic. Here, we elaborate on why both requirements are necessary.

\subsubsection{Why Not Query Complexity Bound?}
In Theorem~\ref{thm:sample-main-intro}, we obtain a polynomial sample complexity for agents with generic value distributions. With query access, such a finite bound is typically not achievable without additional assumptions. For example, consider a seller with $B$ units of a single resource. There are $\frac{3}{2}B$ agents arriving, each consuming 1 unit of the resource when allocated. $\frac{1}{2}B$ agents have a deterministic value of 1, and $B$ agents have a value drawn from a uniform distribution $U[t,t+\delta]$, where $t$ is an unknown constant. The static welfare-optimal price falls in the range $[t,t+\delta]$, as the optimal offline allocation has all agents with value $1$ allocated, and half of the other agents with value around $t$ allocated. However, when $\delta\to0$, it is impossible to find $t$ with a number of queries that does not depend on $\delta$. We leave it as an open question to explore whether it is possible to obtain a polynomial pricing query complexity for static pricing via some definition of ``degree of genericness''. 

\subsubsection{Why Do We Need Generic Distributions?}

While we have designed an anonymous pricing mechanism with nice properties with pricing queries in Theorem~\ref{thm:query-main-intro}, we may ask whether an even simpler static pricing mechanism is possible without the generic assumption. When value distributions are publicly known, setting item prices through the dual optimal solution of the configuration LP can give a $(1-\eps)$ approximation to the optimal welfare. However, even for this known-distribution setting, there is a subtle point that such a \textit{\textbf{static}} pricing mechanism is actually \textit{\textbf{not anonymous}}. How is this even possible?

In fact, to make dual pricing work, an important observation is that the pricing mechanism needs to have an accurate \textit{tie-breaking rule}. That is, for any arriving agent that has two favorite actions with different allocations, the seller has the power to decide which allocation the agent selects. Such a tie-breaking can be discriminatory, making the dual pricing mechanism not truly anonymous. In fact, such an observation has been studied in combinatorial auctions. \cite{cohen2016invisible} points out that in a matching market (where all agents have fixed unit-demand valuations), while a static item pricing (which are usually called Walrasian prices) that clears the market with optimal welfare exists, it is critical that the seller has the tie-breaking power and breaks ties in a coordinated fashion. Suppose the buyers hold the tie-breaking power and can select any utility-maximizing bundle of items at arrival, there exist instances that no static pricing can obtain more than $\frac{2}{3}$ of the optimal welfare in hindsight, and the result holds when each resource has a large supply. \footnote{One simple example is as follows. Consider a seller who wants to sell 3 items with one copy each to 3 unit-demand buyers. Buyer 1 has value $1$ for either item 2 or 3; buyer 2 has value 1 for either item 1 or 3; buyer 3 has value 1 for either item 1 or 2. While a static pricing mechanism with price 1 for each item can obtain the optimal welfare 3, the seller needs to hold the tie-breaking power to decide which item to sell when each agent arrives. In fact, for any static item pricing mechanism, there always exists an arrival order of the buyers such that at most two of the three items are sold. The example can be generalized to the setting where each item has $B$ copies, by having $B$ copies of each buyer.} 
To solve this problem, \cite{cohen2016invisible} proposes an adaptive item pricing mechanism, that adjusts the prices of each item based on the current remaining, and obtains the optimal welfare for the online matching market, and the result can be further generalized to buyers with gross-substitute valuations \cite{berger2020power}. In a more complicated tollbooth setting, where each agent demands an interval on a line graph, \cite{tan2023worst} shows that while the static item pricing can achieve the optimal welfare when the seller has tie-breaking power, it can only obtain a tight $\frac{2}{3}$ fraction of the optimal welfare when the buyers hold the tie-breaking power.

On the other hand, our mechanism for Theorem~\ref{thm:query-main-intro} has a pre-determined set of price vectors, and when each agent arrives, the seller randomly draws a price vector from the set, and neither needs to know the identity of the arriving buyer, nor needs to adjust the price based on the current inventory constraint. Each arriving agent can also select any utility-maximizing action. This means that our mechanism is truly \textit{anonymous}.

\subsection{Related Work}

\vspace{0.5em} \noindent \textbf{Online Resource Allocation through Samples} For the online resource allocation model studied in our paper, \cite{GSW-STOC25} was the first paper to study almost optimal online resource allocation through samples for heterogeneous agents in the large-budget setting. There have been many papers studying online resource allocation with i.i.d. agents or non-identical agents arriving in random order, E.g., \cite{kleinberg2005multiple,DevenurHayes-EC09,AWY-OR14,MR-MOR14,agrawal2014fast,GM-MOR16,KRTV-SICOMP18,devanur2019near,li2022online,bray2025logarithmic}. The main ideas of this series of works are either to use the first $\eps n$ arriving agents as samples to learn an empirical distribution, and then to use a mechanism constructed from the empirical distribution on the rest of the agents; or to run a no-regret learning algorithm on all of the agents and obtain an overall $(1-\eps)$ approximation.

Another line of research studies the online resource allocation problem with regret minimization objectives. \cite{BKMSW-ICML23} and \cite{JLZ-MS25} are the two prior works closest to our setting. Specifically, \cite{BKMSW-ICML23} studies the regret minimization version of the online resource allocation problem with a single resource and non-identical request distributions and designs a static pricing algorithm that requires only a single sample from each request distribution. However, it remains unclear whether their static pricing algorithm can be generalized beyond the single-resource setting. Meanwhile, \cite{JLZ-MS25} investigates the online resource allocation problem with multiple resources. Rather than discussing the sample complexity required to learn a high-quality empirical distribution, \cite{JLZ-MS25} focuses on how the Wasserstein distance between the learned empirical distribution and the true distribution impacts the regret. Additionally, they propose a dynamic pricing algorithm that is not directly comparable to our results, as their algorithm is non-anonymous.

\vspace{0.5em} \noindent \textbf{Online Resource Allocation with Less LP Solving} A distinct line of research focuses on reducing the computational burden of solving large-scale LPs in the online resource allocation problem. While early heuristics based on frequently re-solving the certainty-equivalent LP were shown to achieve small regret in different settings \cite{reiman2008asymptotically,jasin2012resolving,jasin2014reoptimization}, the high computational cost of frequent re-optimization limits their applicability. To address this, \cite{bumpensanti2020resolving} demonstrated that re-solving the LP only a constant number of times is sufficient to maintain a uniformly bounded revenue loss. More recently, literature has challenged the necessity of re-solving altogether. \cite{balseiro2023best} and \cite{he2025online} established that simple primal-dual policies, which dynamically update dual variables based on resource consumption without ever re-solving the LP, can effectively achieve asymptotic optimality. This line of work highlights the effectiveness of low-complexity policies in i.i.d. settings, while our query complexity algorithm does not even solve the LP in a non-i.i.d. environment.

\vspace{0.5em} \noindent \textbf{Sample Complexity of Multi-dimensional Mechanism Design} Since \cite{cole2014sample} introduced sample complexity to the mechanism design literature, there has been much work on learning (approximately) optimal mechanisms. For multi-dimensional mechanism design, the proof of sample complexity follows two very different  techniques. One line of papers \cite{dughmi2014sampling,gonczarowski2021sample,brustle2020multi,GHTZ-COLT21,cai2022computing,teng2025learning} focus on showing that the empirical distribution that is close to the original distribution can be efficiently learned from samples, while the optimal objective of the mechanism design problem satisfies some kind of smoothness with respect to the input distributions. The other line of works \cite{balcan2016sample,morgenstern2016learning,cai2017learning,syrgkanis2017sample,balcan2021much,balcan2023generalization} utilize the pseudo-dimension style analysis to prove uniform convergence bounds for specific mechanism hypotheses classes. The proof of our sample complexity result exactly follows the second line of work.

\vspace{0.5em} \noindent \textbf{Query Complexity} Recently, there has been an increasing attention to learning distributions or mechanisms through pricing / threshold queries \cite{kleinberg2003value,meister2021learning,leme2023pricing,okoroafor2023non,paes2023description, SW-EC24,teng2025learning,chen2025query}. The closest work to our setting is \cite{teng2025learning}, in which they prove a tight bound on learning the optimal item pricing algorithm for unit-demand buyers through pricing queries on single distributions. Our online resource allocation algorithm works on general preference that goes well beyond unit-demand agents. In addition, we use a demand oracle that takes in price vectors instead of querying individual distributions. This also solves an open question from \cite{teng2025learning} on getting reasonable pricing query complexity results with the more natural query access in a multi-dimensional setting.

\vspace{0.5em} \noindent \textbf{Anonymous Pricing for Multi-Dimensional Online Resource Allocation} The anonymous pricing mechanism we have is a special case of Sequential Posted Pricing mechanisms \cite{chawla2010multi}, where the seller has a posted pricing mechanism for each arriving agents. In the multi-dimensional setting, there has been much literature focusing on combinatorial prophet inequality, where the distribution of the arriving agents are known, but the arriving order is unknown. In the single-unit supply setting, it is known that static item pricing achieves a $1/2$-approximation to the optimal welfare in hindsight for agents with XOS valuations \cite{feldman2014combinatorial, DFKL-SICOMP20}, and the result is obtainable through samples \cite{FOCS24-toapper}. For subadditive agents, static item pricing can achieve an $\Omega(1/\log\log m)$ approximation \cite{dutting2020log}, and whether it can achieve constant approximation is still open. Beyond subadditive agents, \cite{chawla2019pricing} shows that static bundle pricing gives an $\Omega(\log \log L/\log L)$ approximation for allocating intervals of length at most $L$. In the multi-unit setting, there are much fewer works in the literature. \cite{chawla2017stability} shows that static item pricing gives $(1-\eps)$ approximation for online interval allocation in the large budget setting. The closest work to our setting is \cite{chawla2025multi}, which gives a non-adaptive anonymous item pricing for XOS-valued agents in a large budget setting. As in our setting they allow the \textit{anonymous} pricing to be \textit{not static}, but they permit the prices to be updated when the remaining budget of each resource changes.

\section{Sample Complexity for Learning Dual Prices}
\label{sec:sample}

In this section, we prove our sample complexity result. To be specific, we show the following:

\begin{Theorem}
    \label{thm:sample-main}
    Let $\epsilon \in (0, 0.5]$ be an error parameter. Consider an online resource allocation instance with underlying generic request distributions $\bD$ and budget $\B$, where a request distribution $\D_i$ is \emph{generic} if for any price vector $\bprice$,  different actions of request $\gamma_i \sim \D_i$ have distinct utilities almost surely, i.e.
    \begin{align*}
        \pr_{\gamma_i = (v_i, \alloc_i, \Theta_i)} \left[\exists \theta_1, \theta_2 \in \Theta_i: \theta_1 \neq \theta_2 \land v_i(\theta_1) - \langle \bprice, \alloc_i(\theta_1) \rangle = v_i(\theta_2) - \langle \bprice, \alloc_i(\theta_2) \rangle  \right] ~=~ 0
    \end{align*}
    
    Let $\opt$ be the objective of \ref{program:ora-ub}$(\bD; \B)$. 
    Given $N =\Omega\left(\frac{m}{\eps^2}\cdot\poly\log(m,n,|\Theta|_{\max})\right)$ samples from $\bD$, there exists a static pricing algorithm that obtains expected total value at least $(1 - O(\epsilon)) \cdot \opt$, provided that $B_j \geq \Omega(\log(nm/\epsilon) \cdot \epsilon^{-2})$ for each resource $j \in [m]$.
\end{Theorem}

\subsection{Proving \Cref{thm:sample-main}}

\vspace{0.5em} \noindent \textbf{Notations.} To prove \Cref{thm:sample-main}, we rely on the following notations:
\begin{itemize}
    \item $\theta^*(\bprice;\gamma_i)$: It represents the best action for request $\gamma_i$ that maximizes the utility, with $\bprice$ being the price vector for each resource, i.e., 
    \[
    \theta^*(\bprice;\gamma_i) ~\in~ \arg \max_{\theta \in \Theta_i} \left(v_i(\theta) - \langle \alloc_i(\theta), \bprice \rangle \right).
    \]
    When there are multiple actions with the same utility, function $\theta^*(\bprice;\gamma_i)$ represents the (possibly random) tie-breaking rule of request $i$, i.e., $\theta^*(\bprice;\gamma_i)$ can be viewed as a sample drawn from a fixed distribution over those actions that maximizes the utility. We will specify the tie-breaking rule we use when using notation $\theta^*(\bprice;\gamma_i)$.
    \item $\wel(\bprice)$ and $\wel_i(\bprice)$: $\wel_i(\bprice)$ represents the expected welfare achieved by satisfying request $\gamma_i$, with $\bprice$ being the price vector for each resource, i.e.,
    \[
    \wel_i(\bprice) ~:=~ \E_{\gamma_i}\left[v_i\big(\theta^*(\bprice; \gamma_i)\big) \right].
    \]
    $\wel(\bprice)$ represents the summation of $\wel_i(\bprice)$, i.e., $\wel(\bprice) = \sum_{i \in [n]} \wel_i(\bprice)$.
    \item $\cons_i(\bprice)$ and $\cons(\bprice)$: $\cons_i(\bprice)$ represents the expected consumption vector of request $\gamma_i$, with $\bprice$ being the price vector for each resource, i.e.,
    \[
    \cons_i(\bprice) = (c_{i, 1} (\bprice), \cdots, c_{i, m}(\bprice)) ~:=~ \E_{\gamma_i}\left[\alloc_i\big(\theta^*(\bprice; \gamma_i)\big) \right].
    \] We further define
    \[
    \cons(\bprice) =  (c_{1} (\bprice), \cdots, c_{m}(\bprice)) ~:=~ \sum_{i \in [n]} \cons_i(\bprice)
    \]
    to be the expected total consumption vector for all requests.
\end{itemize}

\vspace{0.5em} \noindent \textbf{Pre-processing.} Given an online resource allocation instance, we first add $m$ groups of dummy requests into the instance. For $j \in [m]$, the $j$-th group contains $B_j + 1$ dummy requests, such that each request is a unit-demand request which is only interested in taking one unit of resource $j$, with a  valuation independently drawn from $[\delta, 2\delta]$ uniformly at random, where $\delta\to 0$ is an arbitrarily small parameter that does not change the objective of \ref{program:ora-ub}. The main purpose for us to introduce these groups of dummy requests is to guarantee that each resource has enough demand, and therefore the optimal dual prices with respect to either the original distributions or the empirical distributions must be non-zero. Throughout this chapter, we assume that all online resource allocation instances under consideration have undergone this pre-processing step. Consequently, the dual prices corresponding to the optimal solution of \ref{program:dual} with respect to both the original distributions and the empirical distributions are guaranteed to be strictly positive.

\vspace{0.5em} \noindent \textbf{Near-optimal allocation implies near-optimal value.} To show the learnability of the dual prices, a standard approach is to directly analyze the convergence of the optimal dual price vector itself. However, because static pricing can be highly sensitive in an online resource allocation setting, convergence of the dual price vector alone does not guarantee good performance. 

For example, consider a \(B\)-unit Single-Item Prophet Inequality problem, which is a special case of an online resource allocation problem with a single resource and \(n\) unit-demand requests. Assume that each request’s value is independently drawn uniformly at random from an interval \([\ell, r]\), where \(r\) and \(\ell\) are arbitrarily close, i.e., \(r - \ell \to 0\). Let \(\lambda^* \in [\ell, r]\) denote the optimal dual price, and let \(\widehat{\lambda}\) be an estimate of it. When \(r - \ell \to 0\), any convergence bound of the form \(|\lambda^* - \widehat{\lambda}| \le \varepsilon\) for a fixed \(\varepsilon\) fails to provide a performance guarantee for posting the price \(\widehat{\lambda}\). Indeed, it is possible that \(\lambda^* - \varepsilon < \ell\), resulting in a price that may over-allocate the resource, or that \(\lambda^* + \varepsilon > r\), resulting in a price that allocates nothing.

The above observation motivates us to directly consider the convergence of the \emph{expected consumption}, i.e., we wish to find a price vector $\widehat \bprice$ via samples, such that the expected consumption is close but no more than the budget constraint $\B$, i.e., we hope
\begin{align*}
    \cons(\widehat \bprice) ~\in~ \left[(1 - O(\epsilon)) \cdot \B, \B\right].
\end{align*}

We show the feasibility of the above idea by presenting the following \Cref{lma:budget-to-value}, which suggests that if the expected consumption of $\widehat \bprice$ is indeed close to $\B$, the resulting welfare is also sufficiently large.

\begin{Lemma}
    \label{lma:budget-to-value}
    Given an online resource allocation instance with generic distributions, assume for a price vector $\bprice$ and a predefined deterministic tie-breaking rule that only depends on the request,  we have
    \[
    \cons(\bprice) ~=~ \sum_{i \in [n]} \E_{\gamma_i}\left[\alloc_i\big(\theta^*(\bprice; \gamma_i)\big) \right] ~\in~ \left[(1 - 3\epsilon) \cdot \B, (1 - \epsilon) \cdot \B\right],
    \]
    Then, there must be
    \[
    \wel(\bprice)  ~=~ \sum_{i \in [n]} \E_{\gamma_i}\left[v_i\big(\theta^*(\bprice; \gamma_i)\big) \right] ~\geq~ (1 - 3\epsilon) \cdot \opt.
    \]
    Furthermore, the static pricing mechanism with such price vector $\bprice$ achieves an expected reward of at least $(1 - 4\epsilon) \cdot \opt$.
\end{Lemma}

\vspace{0.5em} \noindent \textbf{Uniform convergence for truncated consumption.} Now, we show how to learn a price vector with expected  consumption close to $\B$ via samples. Our core idea is a uniform convergence argument via Pseudo Dimension, similar to the proofs in \cite{FOCS24-toapper} on proving sample complexity for combinatorial prophet inequalities. To be specific, we show the following:

\begin{Lemma}
    \label{lma:uniform-convergence}
    Given an online resource allocation instance with generic request distributions $\D_1, \cdots, \D_n$. For a price vector $\bprice$, let $\cons(\bprice) = (c_1(\bprice), \cdots, c_m(\bprice))$ be the expected consumption of a static pricing mechanism with price vector $\bprice$ and a predefined deterministic tie-breaking rule that only depends on the request, and let $\cons'(\bprice)$ be the expected truncated consumption, such that
    \[
    \cons'(\bprice) ~=~ (c'_1(\bprice), \cdots, c'_m(\bprice)), \text{ where } c'_j(\bprice) = \E_{\bgamma \sim \bD} \left[\min \left\{\sum_{i \in [n]} a_{i,j}\big(\theta^*(\bprice; \gamma_i)\big), 2B_j\right\}\right].
    \]

     For each $i \in [n]$, consider to independently draw $N$ samples $\gamma^{(1)}_i, \cdots, \gamma^{(N)}_i$ and let $\widehat \D_i$ be the empirical distribution over $N$ samples. For a price vector $\bprice$, let $\hat \cons(\bprice) = (\hat c_1(\bprice), \cdots, \hat c_m(\bprice))$ be the expected consumption for the empirical distributions under the same predefined deterministic tie-breaking rule with respect to price vector $\bprice$, and let $\hat \cons'(\bprice) = (\hat c'_1(\bprice), \cdots, \hat c'_m(\bprice))$ be the expected truncated  consumption,  i.e., we have
     \[
     \hat \cons(\bprice) ~=~ \sum_{i \in [n]} \frac{1}{N} \cdot \sum_{k \in [N]} \alloc^{(k)}_i\big(\theta^*(\bprice; \gamma^{(k)}_i)\big), \text{ and } \hat c'_j(\bprice) = \E_{\bgamma \sim \widehat \bD} \left[\min \left\{\sum_{i \in [n]} a_{i,j}\big(\theta^*(\bprice; \gamma_i)\big), 2B_j\right\}\right].
     \]
     Given the assumption that $N \geq C \cdot \frac{m}{\eps^2}\log^2\left(\frac{mn|\Theta|_{\max}}{\eps}\right)$ and $B_j \geq C\cdot \frac{\log(mn/\eps)}{\eps^2}$ for a sufficiently large $C$, with probability at least $1 - \delta$, we have
     \[
   \left \| \frac{\hat \cons'(\bprice)}{2\B} - \frac{\cons'(\bprice)}{2\B} \right \|_{\infty} ~\leq~ \epsilon/4,
     \]
     i.e., for any price vector $\bprice$ and $j \in [m]$, we have $|\hat c'_j(\bprice) - c'_j(\bprice)| \leq \epsilon B_j/2$.
\end{Lemma}

We remark that an easier way to derive a polynomial sample complexity bound is to directly argue the uniform convergence for $\cons(\bprice)$ and $\hat \cons(\bprice)$. However, since $c_j(\bprice)$ and $\hat c_j(\bprice)$, representing the expected consumption of resource $j$ with respect to the original distribution and the empirical distribution respectively, both lie in the range of $[0, n]$, to prove uniform convergence using Pseudo dimension, we must normalize both  $c_j(\bprice)$ and $\hat c_j(\bprice)$ to the interval $[0, 1]$. This implies a super high accuracy requirement for $\hat c_j$ in the form of $|\hat c_j (\bprice) - c_j(\bprice)| \leq O(1)$, resulting in a sample complexity with a polynomial dependency on $n$.

However, for proving \Cref{thm:sample-main}, a much looser accuracy in the form of $|\hat c_j (\bprice) - c_j(\bprice)| \leq O(\epsilon) \cdot B_j$ is already sufficient. To bridge this gap, we introduce the \emph{truncated} expected consumption functions $\cons'(\bprice)$ and $\hat \cons'(\bprice)$. Since both $\cons'(\bprice)$ and $\hat \cons'(\bprice)$ are bounded in $[0, 2B_j]$, \Cref{lma:uniform-convergence} gives an accuracy bound in the form of $|\hat c'_j (\bprice) - c'_j(\bprice)| \leq O(\epsilon) \cdot B_j$ with a sample complexity that does not depend polynomially on $n$. We will further show in the proof of \Cref{thm:sample-main} that the accuracy bound for the truncated consumption can be converted to an accuracy bound for the original consumption function, and therefore a uniform convergence for $\cons'(\bprice)$ and $\hat \cons'(\bprice)$ is sufficient.

\vspace{0.5em} \noindent \textbf{Proving \Cref{thm:sample-main}.} We defer the proofs of \Cref{lma:budget-to-value} and \Cref{lma:uniform-convergence} to \Cref{sec:sample-omitted}, and first prove \Cref{thm:sample-main} via \Cref{lma:budget-to-value} and \Cref{lma:uniform-convergence}.

\begin{proof}[Proof of \Cref{thm:sample-main}]
Let $N = C \cdot \frac{m}{\eps^2}\log^2\left(\frac{mn|\Theta|_{\max}}{\eps}\right)$ for a sufficiently large $C$, and for $i \in [n]$, let $\widehat \D_i$ be the empirical distribution over $N$ samples $\gamma^{(1)}_i, \gamma^{(N)}_i$ independently drawn from $\D_i$. Let $(\widehat \bprice, \{u^*_{i, \gamma_i}\})$ be the optimal solution of \ref{program:dual}$(\widehat \bD; (1 - 2\epsilon) \cdot \B)$, where $\widehat \bD = (\widehat \D_1, \cdots, \widehat \D_n)$, and let $\{\hat x^*_{i, \gamma_i, \theta}\}$ be the optimal solution of \ref{program:ora-ub}$(\widehat \bD; (1 - 2\epsilon) \cdot \B)$. Note that each $\widehat \D_i$ is a discrete distribution with support size at most $N$. Therefore, both \ref{program:ora-ub}$(\widehat \bD; (1 - 2\epsilon) \cdot \B)$ and \ref{program:dual}$(\widehat \bD; (1 - 2\epsilon) \cdot \B)$ are solvable, and the optimal solutions $\{\hat x^*_{i, \gamma_i, \theta}\}$ and $(\widehat \bprice, \{u^*_{i, \gamma_i}\})$ can be derived explicitly.

Our goal is to show that
\begin{align}
    \cons(\widehat \bprice) ~\in~ \left[(1 - 3\epsilon) \cdot \B, (1 - \epsilon) \cdot \B\right], \label{eq:dual-price-converge}
\end{align}
where consumption function $\cons$ is defined on an arbitrary tie-breaking rule independent to the empirical distribution $\widehat \bD$. Then,  \Cref{lma:budget-to-value} guarantees that the static pricing mechanism with $\widehat \bprice$ achieves an expected reward of at least $(1 - 4\epsilon) \cdot \opt$, which finishes the proof of \Cref{thm:sample-main}.

It remains to show \eqref{eq:dual-price-converge}. To achieve this, we follow the following three steps. We first show that the expected consumption $\hat \cons(\widehat \bprice)$ is bounded: we have
\begin{align}
    \hat \cons(\widehat \bprice) ~\in~ \left[(1 - 2.1\epsilon) \cdot \B, (1 - 1.9\epsilon) \cdot \B\right]. \label{eq:dual-price-converge-1}
\end{align}
Next, we show that a bounded expected consumption implies a bounded truncated consumption: we have
\begin{align}
    \hat \cons'(\widehat \bprice) ~\in~ \left[(1 - 2.5\epsilon) \cdot \B, (1 - 1.9\epsilon) \cdot \B\right]. \label{eq:dual-price-converge-2}
\end{align}
Then, applying \Cref{lma:uniform-convergence} gives
\begin{align}
    \cons'(\widehat \bprice) ~\in~ \left[(1 - 3\epsilon) \cdot \B, (1 - 1.4\epsilon) \cdot \B\right]. \label{eq:dual-price-converge-3}
\end{align}
Finally, we show that \eqref{eq:dual-price-converge-3} is sufficient to prove \eqref{eq:dual-price-converge}, which finishes the proof of \Cref{thm:sample-main}.

\vspace{0.5em} \noindent \textbf{Proving \eqref{eq:dual-price-converge-1}.} Note that each $\widehat \D_i$ is a discrete distribution with support size at most $N$. Therefore, both \ref{program:ora-ub}$(\widehat \bD; (1 - 2\epsilon) \cdot \B)$ and \ref{program:dual}$(\widehat \bD; (1 - 2\epsilon) \cdot \B)$ are solvable, and the solution $\{\hat x^*_{i, \gamma_i, \theta}\}$ can be derived explicitly.  Since $\widehat \bprice$ is the optimal solution of the dual, by complementary slackness, for every $x^*_{i, \gamma_i, \theta} > 0$, we have
\[
v_i(\theta) - \langle \widehat \bprice, \alloc_i(\theta) \rangle ~=~u^*_{i, \gamma_i}.
\]
Furthermore, for every $x^*_{i, \gamma_i, \theta'} = 0$, there must be
\[
v_i(\theta') - \langle \widehat \bprice, \alloc_i(\theta) \rangle ~< u^*_{i, \gamma_i},
\]
i.e., when presenting price vector $\widehat \bprice$ to request $\gamma_i$, the utility maximizing action must come from some $\theta \in \Theta_i$ such that $x^*_{i, \gamma_i, \theta} > 0$. 

Now consider the static pricing mechanism with price vector $\widehat \bprice$ for the empirical distribution $\widehat \bD$. Here, we apply a tie-breaking rule that depends on the solution $\{\hat x^*_{i, \gamma_i, \theta}\}$: if the revealed request $i$ equals to $\theta_i \in \text{support}(\widehat \D_i)$, we take action $\theta$ with probability $x^*_{i, \gamma_i, \theta}$. Then, the expected consumption\footnote{Note that notations $\hat \cons(\widehat \bprice)$ or $\theta^*(\widehat \bprice;\gamma_i)$ can't be applied, as we are performing a different tie-breaking rule, instead of the predefined deterministic tie-breaking rule that defines $\hat \cons(\widehat \bprice)$ and $\theta^*(\widehat \bprice;\gamma_i)$.} of the above mechanism is exactly
\begin{align*}
    \sum_{i \in [n]} \E_{\gamma_i \sim \widehat \D_i} \left[x^*_{i, \gamma_i, \theta} \cdot \alloc_i(\theta) \right] ~=~ (1 - 2\epsilon) \cdot \B,
\end{align*}
where the above expression is an equality because pre-processing step guarantees that each resource has enough demand.

Note that the above consumption is using a specifically defined tie-breaking rule, which is different from the predefined deterministic tie-breaking rule that defines notations $\hat \cons(\widehat \bprice)$ and $\theta^*(\widehat \bprice;\gamma_i)$. We now bound $\hat \cons(\widehat \bprice)$ via showing that different tie-breaking rules only bring minor changes to the expected consumption, and therefore $\hat \cons(\widehat \bprice)$ is still close to $ (1 - 2\epsilon) \cdot \B$. 

Recall that for $i \in [n]$, each distribution $\D_i$ is generic. Therefore, all $n \cdot N$ samples $\{\gamma^{(k)}_i\}_{k \in [N], i \in [n]}$ that form the empirical distribution $\widehat \bD$ are in general position almost surely. For a price vector $\bprice$, consider the number of possible ties, i.e., the number of tuples $(k, i): k \in [N], i \in [n]$ such that request $\gamma^{(k)}_i = (v^{(k)}_i, \alloc^{(k)}_i, \Theta^{(k)}_i)$ satisfies the following: there exists $\theta, \theta' \in \Theta^{(k)}_i$ such that 
\[
v^{(k)}_i(\theta) - \langle \bprice, \alloc^{(k)}_i(\theta) \rangle ~=~ v^{(k)}_i(\theta') - \langle \bprice, \alloc^{(k)}_i(\theta') \rangle.
\]
Note that each tie gives the price vector a linear constraint. Since $\bprice$ is an $m$-dimensional vector and all $n \cdot N$ samples are in general position almost surely, the total number of ties is at most $m$ almost surely. Then, price vector $\widehat \bprice$ can meet at most $m$ ties when facing requests from $\widehat \bD$. Since the  consumption changes by $O(1)$ when the predefined deterministic tie-breaking rule directs a request to a different action from the rule of following $\{\hat x^*_{i, \gamma_i, \theta}\}$, there must be
\[
\big|\hat c_j(\widehat \bprice) - (1 - 2\epsilon) \cdot B_j\big| ~\leq~ \frac{m}{N} ~\leq~ 1 ~\leq~ 0.1\epsilon \cdot B_j,
\]
where the first inequality follows the fact that each sample $\gamma^{(k)}_i$ leading to a tie shows up with probability $|\text{support}(\widehat \D_i)|^{-1} = N^{-1}$, and the last inequality holds when $B_j \geq 10\epsilon^{-2}$. Therefore, we have
\[
\hat \cons(\widehat \bprice) ~\in~ [(1 - 2.1\epsilon) \cdot \B, (1 -1.9\epsilon) \cdot \B ].
\]

\vspace{0.5em} \noindent \textbf{Proving \eqref{eq:dual-price-converge-2} via \eqref{eq:dual-price-converge-1}.}
Recall that $\hat \cons(\widehat \bprice)$ represents the expected consumption of $\widehat \bprice$ with respect to the empirical distribution $\widehat \bD$, and $\hat \cons'(\widehat \bprice)$ represents the expected truncated consumption. Therefore,
\[
\hat \cons'(\widehat \bprice) ~\leq~ \hat \cons(\widehat \bprice) ~\leq~ (1 -1.9\epsilon) \cdot B_j 
\]
immediately holds. It remains to show the lower bound of \eqref{eq:dual-price-converge-2}, which is sufficient to show that for every $j \in [m]$, we have
\begin{align}
    \hat c'_j(\widehat \bprice) ~\geq~ \hat c_j(\widehat \bprice) - 0.4 \epsilon \cdot B_j, \label{eq:dual-price-converge-mid}
\end{align}
provided that $\hat c_j(\widehat \bprice) \leq (1 - 1.9\epsilon) \cdot B_j \leq B_j$. 

Fix $j$. For $i \in [n]$, let random variable
\begin{align*}
    \xi_i ~:=~ a_{i,j}\big(\theta^*(\widehat \bprice; \gamma_i) \big)~:~ \gamma_i = (v_i, \alloc_i, \Theta_i) ~\sim \widehat \D_i
\end{align*}
represent the random consumption of request $i$ when facing price vector $\widehat \bprice$. By the definition of $\xi_i$, we have
\[
\hat c_j(\widehat \bprice) ~=~ \E \left[\sum_{i \in [n]} \xi_i\right] \qquad \text{and} \qquad \hat c'_j(\widehat \bprice) ~=~ \E \left[\min \left\{\sum_{i \in [n]} \xi_i, 2B_j\right\}\right].
\]
Note that $\xi_i \in [0, 1]$, and therefore $\sum_{i \in [n]} \xi_i \leq n$ always holds. If $n \leq 2B_j$, then  $\hat c_j(\widehat \bprice) = \hat c'_j(\widehat \bprice)$ must hold. Otherwise, the definition of  $\hat c_j(\widehat \bprice) $ and $\hat  c'_j(\widehat \bprice)$ gives
\begin{align*}
    \hat c'_j(\widehat \bprice) ~&=~ \E \left[\sum_{i \in [n]} \xi_i ~\Big | ~ \sum_{i \in [n]} \xi_i \leq 2B_j\right] \cdot \pr \left[\sum_{i \in [n]} \xi_i \leq 2B_j \right] + 2B_j \cdot \pr \left[\sum_{i \in [n]} \xi_i > 2B_j \right] \\
    ~&=~\E \left[\sum_{i \in [n]} \xi_i ~\Big | ~ \sum_{i \in [n]} \xi_i \leq 2B_j\right] \cdot \pr \left[\sum_{i \in [n]} \xi_i \leq 2B_j \right] + (2B_j - n + n) \cdot \pr \left[\sum_{i \in [n]} \xi_i > 2B_j \right] \\
    ~&\geq~ \hat c_j(\widehat \bprice) + (2B_j - n) \cdot \pr \left[\sum_{i \in [n]} \xi_i > 2B_j \right] ~\geq~ \hat c_j(\widehat \bprice) - n \cdot \pr \left[\sum_{i \in [n]} \xi_i > 2B_j \right].
\end{align*}
where the first inequality uses the fact that $\sum_{i \in [n]} \xi_i \leq n$ always holds. Therefore, to show \eqref{eq:dual-price-converge-mid}, it remains to show 
\[
\pr \left[\sum_{i \in [n]} \xi_i > 2B_j \right] ~\leq~ \frac{1}{n} ~\leq~ \frac{0.4 \epsilon B_j}{n},
\]
where the second inequality holds when $B_j \geq 3\epsilon^{-2}$. To show the first inequality, note that $\{\xi_i - \E[\xi_i]\}$ is a set of mean-zero random variables bounded in $[-1, 1]$, such that
\[
\var \left(\sum_{i \in [n]} \xi_i - \E[\xi_i]\right) ~=~ \sum_{i \in [n]} \var \left( \xi_i \right) ~\leq~ \sum_{i \in [n]} \E[\xi_i^2] ~\leq~ \hat c_j(\bprice) ~\leq~ B_j.
\]
Then, by Bernstein's Inequality (\Cref{Bernstein-bounded}), we have
\begin{align*}
    \pr \left[\sum_{i \in [n]} \xi_i > 2B_j \right] ~\leq~ \pr \left[\sum_{i \in [n]} \xi_i - \E[\xi_i] > \epsilon \cdot B_j \right] ~\leq~ \exp \left(- \frac{\epsilon^2 B^2_j/2}{ B_j + \epsilon B_j/3}\right) ~\leq~ \frac{1}{n},
\end{align*}
where the first inequality uses the fact that $\sum_{i \in [n]} \E[\xi_i] = \hat c_j(\bprice) \leq B_j$, and the last inequality holds when $B_j \geq 3\log n \cdot \epsilon^{-2}$.

\vspace{0.5em} \noindent \textbf{Proving \eqref{eq:dual-price-converge} via \eqref{eq:dual-price-converge-3}.}
Recall that $\cons(\widehat \bprice)$ represents the expected consumption of $\widehat \bprice$ with respect to the original distribution $\bD$, and $\cons'(\widehat \bprice)$ represents the expected truncated consumption. Therefore,
\[
\cons(\widehat \bprice) ~\geq~ \cons'(\widehat \bprice) ~\geq~ (1 -3\epsilon) \cdot B_j 
\]
immediately holds. It remains to show the upper bound of \eqref{eq:dual-price-converge}, which is sufficient to show that for every $j \in [m]$, we have
\begin{align}
    c_j(\widehat \bprice) ~\leq~ (1 - \epsilon) \cdot B_j, \label{eq:dual-price-converge-mid-2}
\end{align}
provided that $c'_j(\widehat \bprice) \leq (1 - 1.4\epsilon) \cdot B_j$. 

We prove via contradiction. Fix $j$. Assume $c_j(\widehat \bprice) > (1 - \epsilon) \cdot B_j$. For $i \in [n]$, let random variable 
\begin{align*}
    \eta_i ~:=~ a_{i,j}\big(\theta^*(\widehat \bprice; \gamma_i) \big)~:~ \gamma_i = (v_i, \alloc_i, \Theta_i) ~\sim \D_i
\end{align*}
represent the random consumption of request $i$ when facing price vector $\widehat \bprice$. Note that $\{\eta_i - \E[\eta_i]\}$ is a set of mean-zero random variables bounded in $[-1, 1]$, such that
\[
\var \left(\sum_{i \in [n]} \eta_i - \E[\eta_i]\right) ~=~ \sum_{i \in [n]} \var \left(\eta_i \right) ~\leq~ \sum_{i \in [n]} \E[\eta_i^2] ~\leq~ c_j(\bprice).
\]
Then, by Bernstein's Inequality (\Cref{Bernstein-bounded}), we have
\begin{align*}
    \pr \left[\sum_{i \in [n]} \eta_i < (1 - 1.1\epsilon) \cdot B_j \right] ~&\leq~ \pr \left[\sum_{i \in [n]} \eta_i - c_j(\bprice) <  -0.1 \epsilon \cdot c_j(\bprice) \right] \\
    ~&\leq~ \exp \left(- \frac{0.01\epsilon^2 \cdot c^2_j(\bprice)/2}{ c_j(\bprice) + 0.1\epsilon \cdot c_j(\bprice)/3}\right) ~\leq~ 0.1\epsilon,
\end{align*}
where the last inequality holds when $\epsilon \leq 0.5$ and $B_j \geq 1000 \log \epsilon^{-1} \cdot \epsilon^{-2}$. Then, the definition of $c'_j(\widehat \bprice)$ gives
\begin{align*}
    c'_j(\widehat \bprice) ~&\geq~ (1 - 1.1\epsilon) \cdot B_j \cdot \pr \left[\sum_{i \in [n]} \eta_i \geq (1 - 1.1\epsilon) \cdot B_j \right] + 0 \cdot \pr \left[\sum_{i \in [n]} \eta_i < (1 - 1.1\epsilon) \cdot B_j \right] \\
    ~&\geq~ (1 - 1.1\epsilon) \cdot B_j \cdot (1 - 0.1\epsilon) ~\geq~ (1 - 1.2\epsilon) \cdot B_j,
\end{align*}
which is in contrast to the assumption that $c'_j(\widehat \bprice) \leq (1 - 1.4\epsilon) \cdot B_j$. Therefore, $c_j(\widehat \bprice) \leq (1 - \epsilon) \cdot B_j$ must hold.
\end{proof}

\section{Pricing Query Complexity for Online Resource Allocation}
\label{sec:query-complexity}

In this section, we prove our main pricing query complexity result. To be specific, we show the following:

\begin{Theorem}
    \label{thm:main-formal}
    Let $\epsilon \in (0, 0.5]$ be an error parameter. For an online resource allocation instance with underlying request distributions $\bD$ and budget $\B$, let $\opt$ be the objective of \ref{program:ora-ub}$(\bD; \B)$. Given an estimate $\widehat \opt$ that satisfies $\widehat \opt \leq \opt \leq \beta \cdot \widehat \opt$ and $T \geq n$ queries from each request distribution $\D_i$, there exists an anonymous pricing algorithm that obtains expected total value at least $(1 - O(\epsilon)) \cdot \opt$, provided that $B_j \geq \Omega(\log(nm\beta/\epsilon) \cdot \epsilon^{-2})$ for each resource $j \in [m]$.
\end{Theorem}

We remark that \Cref{thm:main-formal} requires an estimate $\widehat{\opt}$ of the optimal value~$\opt$. 
Such an estimate is necessary because it provides the \emph{scale} of the problem: without any information about the scale, it is impossible to determine, using only a polynomial number of queries, whether the values in the underlying online resource allocation instance are on the order of $10^{-100}$ or $10^{100}$. To give a more concrete intuition, consider a simple online resource allocation instance with single resource having $B$ supplies. There are $n = 1.5 B$ requests, where each request requires one unit of resource. $B$ requests are with value $w = 10^k$ for some arbitrary $k \in \mathbb{R}$, while the remaining $0.5 B$ requests are with value $10w$. The arrival order of requests are random. Then, as the optimal algorithm needs to save $0.5 B$ units of the resource for high-value requests, it's necessary for the algorithm to determine the value of $w$ (up to a constant factor). However, as $w = 10^k$ while $k$ can be some arbitrary real number, it's impossible for the algorithm to determine $w$ using $\poly(B)$ pricing queries with only comparison feedback. Therefore,  an estimate of $\opt$ is necessary.

Importantly, only a coarse approximation of $\opt$ is needed. 
If the estimate $\widehat{\opt}$ differs from $\opt$ by a multiplicative factor of~$\beta$, this inaccuracy results in only an additional $\log \beta$ factor in our large-budget requirement.

We show that \Cref{alg:exp} is the desired algorithm for \Cref{thm:main-formal}. As discussed in \Cref{sec:technical-overview}, the learning phase of \Cref{alg:exp} is a $T$-rounds learning algorithm that can be viewed as a large online resource allocation problem. We use the exponential pricing algorithm to update prices between rounds. In the auction phase, we collect all the price vectors appeared in the learning phase as set $\Lambda$, and present to each newly arriving request a price vector independently drawn from $\Lambda$.

\begin{algorithm}[tbh]
\caption{\textsc{Exponential Pricing via Estimates}}
\label{alg:exp}
\KwIn{ instance $\mathcal{I}$, budgets $\{B_j\}_{j \in [m]}$, error parameter $\epsilon$, estimates $\esopt$}
\tcp{Learning Phase}
Initialize $\bprice^{(1)}$: for $j \in [m]$, we set $\price^{(1)}_j = \initprice$, where $\initprice \gets \esopt \cdot \frac{\epsilon^2}{m}$. \\
\For{$t = 1, \ldots, T$}
{
    \For{$i = 1, \ldots n$}
    {
        $\gamma^{(t)}_i = (\val^{(t)}_i, a^{(t)}_i, \Theta^{(t)}_i) \gets $ request $i$ in $t$-th round. \\
        $\theta^{(t)}_i = \arg\max_{\theta \in \Theta^{(t)}_i} \left( \val^{(t)}_i(\theta) - \langle\bprice^{(t)}, a^{(t)}_i(\theta) \rangle \right)$.
    }
Let $\realallocit = (\realallocbase_{i,1}, \cdots, \realallocbase_{i, m})$ be the resource vector consumed by request $\gamma^{(t)}_i$. \\
Let $\realalloct = (\realallocbase_{1}, \cdots, \realallocbase_{m}) := \sum_{i \in [n]} \realallocit$ be the total resource consumed in $t$-th round. \\
Update $\bprice^{(t+1)}$: for $j \in [m]$, we set
\[
\textstyle \price^{(t+1)}_j ~:=~ \price^{(t)}_j \cdot \exp\left(\frac{\eps}{n} \cdot \big(\realallocbase_{j} - (1 - \epsilon) \cdot B_j\big)\right)
\]
}
\tcp{Auction Phase}
Let $\Lambda=\{\bprice^{(t)}\}_{t \in [T]}$. \\
\For{$n$ {requests arriving in arbitrary order}}
{
    Choose a price vector from $\Lambda$ uniformly at random, and propose to the arriving request. \\
    The arriving request performs \emph{best response} against the price vector. \\
    \textbf{Terminate} if any resource $j$ exceeds budget $B_j - 1$
}
\end{algorithm}

\subsection{Notations.} To prove \Cref{thm:main-formal}, we will reuse some notations used in the proof of \Cref{thm:sample-main}, which are briefly restated below for convenience:
\begin{itemize}
    \item $\theta^*(\bprice;\gamma_i)$: It represents the best action for request $\gamma_i$ that maximizes the utility, with $\bprice$ being the price vector for each resource, i.e., 
    \[
    \theta^*(\bprice;\gamma_i) ~\in~ \arg \max_{\theta \in \Theta_i} \left(v_i(\theta) - \langle \alloc_i(\theta), \bprice \rangle \right).
    \]
    When there are multiple actions with the same utility, we assume in this section that each request applies an arbitrary (and possibly random) but consistent tie-breaking rule.
    \item $\wel(\bprice)$ and $\wel_i(\bprice)$: $\wel(\bprice)$ $\wel_i(\bprice)$ represents the expected welfare achieved by satisfying all requests and request $\gamma_i$ respectively, with $\bprice$ being the price vector for each resource, i.e.,
    \[
    \wel_i(\bprice) ~:=~ \E_{\gamma_i}\left[v_i\big(\theta^*(\bprice; \gamma_i)\big) \right], \quad \text{and} \quad \wel(\bprice) ~:=~ \sum_{i \in [n]} \wel_i(\bprice).
    \]
    \item $\cons_i(\bprice)$ and $\cons(\bprice)$: $\cons_i(\bprice)$ and $\cons(\bprice)$ represents the expected consumption vector of all requests/request $\gamma_i$ respectively, with $\bprice$ being the price vector for each resource, i.e.,
    \[
    \cons_i(\bprice)  ~:=~ \E_{\gamma_i}\left[\alloc_i\big(\theta^*(\bprice; \gamma_i)\big)\right], \quad \text{and} \quad \cons(\bprice) ~:=~ \sum_{i \in [n]} \cons_i(\bprice).
    \]
\end{itemize}

We also introduce the following new notations:
\begin{itemize}
    \item $\{x^*_{i, \gamma_i, \theta} \}$. We use $x^*_{i, \gamma_i, \theta}$ to represent the optimal solution of \ref{program:ora-ub}$(\bD; \B)$. We define $\{x^*_{i, \gamma_i, \theta} \}$ for analytical purposes only, without the need to actually solve for the values.
    \item $\alg(\Lambda)$: It represents the expected welfare gained in the auction phase of \Cref{alg:exp}, assuming the set of price vector is $\Lambda$. To prove \Cref{thm:main-formal}, it's equivalent to show that 
    \[
    \E_{\Lambda}[\alg(\Lambda)] ~\geq~ (1 - O(\epsilon)) \cdot \opt.
    \]
    We remark that the exact value of $\alg(\Lambda)$ may also depend on the arrival order of requests in the auction phase. We will later show that the arrival order does not affect our analysis of lower bounding $\alg(\Lambda)$, and for simplicity, we assume $\alg(\Lambda)$ represents the expected reward achieved in the worst possible arrival order.

    \item $\algt_i(\Lambda)$ and $\algt(\Lambda)$: We define
    \[
    \algt_i(\Lambda) ~:=~ \sum_{t \in [T]} \wel_i(\bprice^{(t)}) \qquad \text{and} \qquad \algt(\Lambda) ~:=~ \sum_{i \in [n]} \algt_i(\Lambda)
    \]
    to be the sum of expected welfare gained via playing price vectors in $\Lambda$ from request $i$ and all requests, respectively.
    \item $\pbe(\Lambda)$: We define $\pbe(\Lambda)$, where $\be$ is short for ``budget exhaustion'', to be the probability that the auction phase of \Cref{alg:exp} exhaust the budget when the set of price vectors is $\Lambda$, i.e., assuming we don't execute the ``Terminate'' step in Line 14 of \Cref{alg:exp} and let the algorithm go through all requests, $\pbe(\Lambda)$ represents the probability that there exists $j \in [m]$ such that the auction phase allocated at least $B_j$ unit of resource $j$.
\end{itemize}

\subsection{Proof of \Cref{thm:main-formal}}

Now, we provide a formal proof of \Cref{thm:main-formal}. Recall that we aim at showing
\[
    \E_{\Lambda}[\alg(\Lambda)] ~\geq~ (1 - O(\epsilon)) \cdot \opt.
\]
We first lower-bound $\alg(\Lambda)$ as
\begin{align}
    \label{eq:alg-lower}
    \alg(\Lambda) ~=~ \sum_{i \in [n]}\pr\left[\text{algorithm reaches request }i\right] \cdot \frac{1}{T} \cdot \algt_i(\Lambda) ~\geq~ \frac{1 - \pbe(\Lambda)}{T} \cdot \algt(\Lambda),
\end{align}
where the inequality follows from the fact that the probability that the auction phase reaches request $i$ is at least the probability that the auction phase never exhaust the budget, which is exactly $1 - \pbe(\Lambda)$. Therefore, to prove \Cref{thm:main-formal}, it's equivalent to show that
\begin{align*}
    \E_{\Lambda} \left[(1 - \pbe(\Lambda)) \cdot \algt(\Lambda) \right] ~\geq~ (1 - O(\epsilon)) \cdot T \cdot \opt.
\end{align*}

\vspace{0.5em} \noindent \textbf{Bounding $\algt(\Lambda)$.} Fix $\Lambda$. We first give a lower bound for $\wel_i(\bpricet)$. To achieve this, we define the following randomized action $\widehat \theta_i$, which only depends on the type of request $i$:
\begin{align*}
    \widehat \theta_i ~:=~  \begin{cases} \phi & \text{with probability } \epsilon  \\ \theta \in \Theta_{i} & \text{with probability } x^*_{i,\gamma_i, \theta} \cdot (1 - \epsilon), \text{ assuming } \gamma_i \text{ is the true realization}
\end{cases}~,
\end{align*}
i.e., we ask $\widehat \theta_i$ to take the null action with probability $\epsilon$, and follow the optimal solution of \eqref{program:ora-ub} otherwise. To bound $\wel_i(\bpricet)$, the main idea is that the utility achieved by playing action $\widehat \theta_i$ can't be more than playing $\theta^*(\bpricet, \gamma_i)$. Then, we have
\begin{align}
    \wel_i(\bpricet) ~&=~ \E_{\gamma_i}\left[v_i\big(\theta^*(\bpricet; \gamma_i) \big) - \left \langle \alloc_i\big(\theta^*(\bpricet; \gamma_i)\big), \bpricet \right \rangle  \right] +  \E_{\gamma_i}\left[ \left \langle \alloc_i\big(\theta^*(\bpricet; \gamma_i)\big), \bpricet \right\rangle  \right] \notag \\
    ~&\geq~ \E_{\gamma_i}\left[v_i\big(\widehat \theta_i \big) - \left \langle \alloc_i\big(\widehat \theta_i\big), \bpricet \right \rangle  \right] +  \E_{\gamma_i}\left[ \left \langle \alloc_i\big(\theta^*(\bpricet; \gamma_i)\big), \bpricet \right\rangle  \right] \notag \\
    ~&=~ \E_{\gamma_i}\left[v_i\big(\widehat \theta_i \big) \right] - \E_{\gamma_i} \left[  \left \langle \alloc_i\big(\widehat \theta_i\big), \bpricet \right \rangle  \right] +  \E_{\gamma_i}\left[ \left \langle \alloc_i\big(\theta^*(\bpricet; \gamma_i)\big), \bpricet \right\rangle  \right].\label{eq:wel-lower-1}
\end{align}

To further simplify \eqref{eq:wel-lower-1}, for the first term in \eqref{eq:wel-lower-1}, we have
\[
\E_{\gamma_i}\left[v_i\big(\widehat \theta_i \big) \right] ~=~ (1 - \epsilon) \cdot \E_{\gamma_i} \left[\sum_{\theta \in \Theta_i} x^*_{i, \gamma_i, \theta} \cdot v_i(\theta)\right].
\]

For the second term in \eqref{eq:wel-lower-1}, we have
\begin{align*}
    \E_{\gamma_i} \left[  \left \langle \alloc_i\big(\widehat \theta_i\big), \bpricet \right \rangle  \right] ~&=~ \left \langle \E_{\gamma_i} \left[ \sum_{\theta \in \Theta_{i}}(1 - \epsilon) \cdot x^*_{i, \gamma_i, \theta} \cdot \alloc_{i} (\theta) \right], \bpricet \right \rangle \\
    ~&=~ \left \langle (1 - \epsilon) \cdot \E_{\gamma_i} \left[ \sum_{\theta \in \Theta_{i}}  x^*_{i, \gamma_i, \theta} \cdot  \alloc_{i} (\theta) \right], \bpricet \right \rangle.
\end{align*}

For the third term in \eqref{eq:wel-lower-1}, we have
\[
\E_{\gamma_i}\left[ \left \langle \alloc_i\big(\theta^*(\bpricet; \gamma_i)\big), \bpricet \right\rangle  \right] ~=~ \left \langle \E_{\gamma_i}\left[ \alloc_i\big(\theta^*(\bpricet; \gamma_i)\big) \right], \bpricet \right\rangle ~=~ \left \langle \cons_i(\bpricet), \bpricet \right\rangle.
\]
Applying the above three equalities to \eqref{eq:wel-lower-1}, we have

\begin{align*}
    \wel_i(\bpricet) ~\geq~ (1 - \epsilon) \cdot \E_{\gamma_i} \left[ \sum_{\theta \in \Theta_{i}} x^*_{i, \gamma_i, \theta} \cdot v_{i}(\theta) \right] + \left \langle \cons_i(\bpricet) - (1 - \epsilon) \cdot \E_{\gamma_i} \left[ \sum_{\theta \in \Theta_{i}}  x^*_{i, \gamma_i, \theta} \cdot  \alloc_{i} (\theta) \right], \bpricet \right\rangle.
\end{align*}

Now summing the above inequality over $i \in [n]$ and $t \in [T]$, we have

\begin{align}
    \algt(\Lambda) ~=~& \sum_{t \in [T]} \sum_{i \in [n]} \wel_i(\bpricet) \notag \\
    ~\geq~& (1 - \epsilon) \cdot T \cdot \sum_{i \in [n]} \E_{\gamma_i} \left[ \sum_{\theta \in \Theta_{i}} x^*_{i, \gamma_i, \theta} \cdot v_{i}(\theta) \right] \notag \\
    &+ \sum_{t \in [T]} \left \langle  \cons(\bpricet) - (1 - \epsilon) \cdot  \sum_{i \in [n]} \E_{\gamma_i} \left[ \sum_{\theta \in \Theta_{i}}  x^*_{i, \gamma_i, \theta} \cdot  \alloc_{i} (\theta) \right], \bpricet \right\rangle \notag \\
    ~\geq~& (1 - \epsilon) \cdot T \cdot \opt + \sum_{t \in [T]} \left \langle  \cons(\bpricet) - (1 - \epsilon) \cdot  \B, \bpricet \right\rangle \notag \\
    ~=~& (1 - \epsilon) \cdot T \cdot \opt \label{eq:first}\\
    &+ \sum_{t \in [T]} \left \langle  \realalloct - (1 - \epsilon) \cdot  \B, \bpricet \right\rangle \label{eq:second}\\
    &+ \sum_{t \in [T]} \left \langle  \cons(\bpricet) - \realalloct , \bpricet \right\rangle. \label{eq:third}
\end{align}

Recall that our goal is to show
\begin{align*}
    \E_{\Lambda} \left[(1 - \pbe(\Lambda)) \cdot \algt(\Lambda) \right] ~\geq~ (1 - O(\epsilon)) \cdot T \cdot \opt.
\end{align*}
Then, it's sufficient to separately lower bound the expectation of $(1 - \pbe(\Lambda)) \cdot \eqref{eq:first}$, $(1 - \pbe(\Lambda)) \cdot \eqref{eq:second}$, and $(1 - \pbe(\Lambda)) \cdot \eqref{eq:third}$.

\vspace{0.5em} \noindent \textbf{Bounding $\E_{\Lambda}[(1 - \pbe(\Lambda)) \cdot \eqref{eq:first}]$.}  This is the relatively easier case. We have
\begin{align}
    \E_{\Lambda} \left[(1 - \pbe(\Lambda)) \cdot (1 - \epsilon) \cdot T \cdot \opt\right] ~\geq~ (1 - \epsilon) \cdot T \cdot \opt - T \cdot \opt \cdot \E_{\Lambda} \left[\pbe(\Lambda)\right]. \label{eq:first-1}
\end{align}

It remains to bound $\E_{\Lambda} \left[\pbe(\Lambda)\right]$. We use the following \Cref{lma:pbesmall} to bound $\E_{\Lambda} \left[\pbe(\Lambda)\right]$:

\begin{restatable}{Lemma}{pbesmall}
    \label{lma:pbesmall}
    For \Cref{alg:exp}, we have $\E_{\Lambda} \left[\pbe(\Lambda)\right] \leq 3 \epsilon$.
\end{restatable}

The proof of \Cref{lma:pbesmall} requires arguing the concentration of a martingale process. We defer the proof of \Cref{lma:pbesmall} to Section~\ref{subsec:proof-lma-pbesmall}. Then, applying \Cref{lma:pbesmall} to \eqref{eq:first-1} gives
\begin{align}
    \E_{\Lambda} \left[(1 - \pbe(\Lambda)) \cdot (1 - \epsilon) \cdot T \cdot \opt\right] ~\geq~ (1 - 4 \epsilon) \cdot T \cdot \opt. \label{eq:first-final}
\end{align}

\vspace{0.5em} \noindent \textbf{Bounding $\E_{\Lambda}[(1 - \pbe(\Lambda)) \cdot \eqref{eq:second}]$.} Before taking the expectation for $(1 - \pbe(\Lambda)) \cdot \eqref{eq:second}$, we further apply the idea of ``no-regret w.r.t. zero vector'' from \cite{GSW-STOC25} to simplify \eqref{eq:second}. To be specific, we rely on the following lemma:

\begin{Lemma}[Lemma 2.7 in \cite{GSW-STOC25}] \label{lma:noRegret}
Consider an online setting where in the $t$-th step we play a price $\lambda_{t, j} \geq 0$ for  item $j$. On playing price $\lambda_{t, j}$, we get revenue  $\lambda_{t, j} \cdot r_{t,j}$ for some adversarially chosen $r_{t,j} \in [-1,1]$ (denote $r_{t,j}^+ := \max\{0, r_{t, j}\}$). Then, playing exponential prices   $\lambda_{t, j} = \initprice \cdot \exp\left(\delta \cdot \sum_{\ell \, < \, t}  r_{\ell,j} \right)$ for some $\delta \in (0,1/2)$ and $\initprice >0$ gives  us total revenue 
\[ 
         \sum_{t = 1}^T \lambda_{t, j} \cdot r_{t,j} ~\geq ~ - \frac{2\initprice}{\delta} - 4\delta \cdot \sum_{t = 1}^T \lambda_{t, j} \cdot r_{t,j}^+ .
        \] 
\end{Lemma}

For every $j \in [m]$, we apply \Cref{lma:noRegret} with $r_{t, j} = \frac{1}{n} \cdot (\realallocbase_j - (1 - \epsilon) \cdot B_j)$ and $\delta = \epsilon$, where the required condition $r_{t,j} \in [-1, 1]$ follows from the fact that it's impossible to consume more than $n$ unit of resource $j$ in one round, and $B_j \leq n$. Then, \Cref{lma:noRegret} is applicable by noting that $\lambda_{t,j} = \initprice \cdot \exp \left(\epsilon \cdot \sum_{\ell < t} \frac{1}{n} \cdot (\realallocbase_j - (1 - \epsilon) \cdot B_j)\right)$, and the inequality in \Cref{lma:noRegret} gives 
\[
\sum_{t \in [T]} \lambda^{(t)}_j \cdot (\realallocbase_j - (1 - \epsilon) \cdot B_j) ~\geq~ -\frac{2n \initprice}{\epsilon} - 4\epsilon \cdot \sum_{t \in [T]}  \lambda_{t, j} \cdot (\realallocbase_j - (1 - \epsilon) \cdot B_j)^+.
\]

Summing the above inequality for all $j \in [m]$, we have
\begin{align}
    \sum_{t \in [T]} \left \langle  \realalloct - (1 - \epsilon) \cdot  \B, \bpricet \right\rangle ~\geq~ -\frac{2nm \initprice}{\epsilon} - 4\epsilon \cdot \sum_{t \in [T]} \left \langle  \realalloct , \bpricet \right\rangle, \label{eq:furtheruse}
\end{align}
where we also use the fact that $(\realallocbase_j - (1 - \epsilon) \cdot B_j)^+ \leq \realallocbase_j$ to simplify the right hand side. Now we multiply $1 - \pbe(\Lambda)$ to both sides of \eqref{eq:furtheruse}, which gives
\begin{align*}
  &(1 - \pbe(\Lambda)) \cdot  \sum_{t \in [T]} \left \langle  \realalloct - (1 - \epsilon) \cdot  \B, \bpricet \right\rangle  \\
  ~\geq~& - (1 - \pbe(\Lambda)) \cdot \left(\frac{2nm \initprice}{\epsilon} + 4\epsilon \cdot \sum_{t \in [T]} \left \langle  \realalloct , \bpricet \right\rangle\right)\\
  ~\geq~&  - \frac{2nm \initprice}{\epsilon} - 4\epsilon \cdot \sum_{t \in [T]} \left \langle  \realalloct , \bpricet \right\rangle.
\end{align*}
Taking the expectation over $\Lambda$, we have
\begin{align}
    \E_{\Lambda} \left[(1 - \pbe(\Lambda)) \cdot  \sum_{t \in [T]} \left \langle  \realalloct - (1 - \epsilon) \cdot  \B, \bpricet \right\rangle\right] ~&\geq~ - \frac{2nm \initprice}{\epsilon} - 4\epsilon \cdot \E_{\Lambda} \left[\sum_{t \in [T]} \left \langle  \realalloct , \bpricet \right\rangle \right] \notag \\
    ~&=~ - \frac{2nm \initprice}{\epsilon} - 4\epsilon \cdot \E_{\Lambda} \left[\sum_{t \in [T]} \left \langle  \cons(\bpricet) , \bpricet \right\rangle \right] \notag \\
    ~&\geq~ - \frac{2nm \initprice}{\epsilon} - 4\epsilon \cdot \E_{\Lambda} \left[\algt(\Lambda) \right], \label{eq:second-1}
\end{align}
where the second line uses the fact that when $\bpricet$ is fixed, the realization of $\realalloct$, which only depends on samples $\{\gamma^{(t)}_i\}_{i \in [n]}$, is a random vector with $\cons(\bpricet)$ being its expectation, and the last line uses the fact that $\left \langle  \cons(\bpricet) , \bpricet \right\rangle$ represents the expected total revenue collected by presenting price vector $\bpricet$ to every request $i \in [n]$, which is further upper-bounded by the total welfare collected by presenting $\bpricet$ to every request $i \in [n]$.

It remains to upper-bound $\E_{\Lambda} \left[\algt(\Lambda) \right]$. To achieve this, we rely on the following \Cref{lma:algtsmall}:

\begin{restatable}{Lemma}{algtsmall}
    \label{lma:algtsmall} 
    For \Cref{alg:exp}, we have $\E_{\Lambda} \left[\algt(\Lambda) \right] \leq 2T \cdot \opt$.
\end{restatable}

The proof of \Cref{lma:algtsmall} is deferred to Section~\ref{subsec:proof-lma-algtsmall}. Applying \Cref{lma:algtsmall} to \eqref{eq:second-1} gives
\begin{align}
    \E_{\Lambda} \left[(1 - \pbe(\Lambda)) \cdot  \sum_{t \in [T]} \left \langle  \realalloct - (1 - \epsilon) \cdot  \B, \bpricet \right\rangle\right] ~\geq~ - \frac{2nm \initprice}{\epsilon} - 8\epsilon \cdot T \cdot  \opt. \label{eq:second-final}
\end{align}

\vspace{0.5em} \noindent \textbf{Bounding $\E_{\Lambda}[(1 - \pbe(\Lambda)) \cdot \eqref{eq:third}]$.} This is the case where the correlation issue arises. Note that if we simply want to lower bound $\E_{\Lambda}\left[\sum_{t \in [T]} \left \langle  \cons(\bpricet) - \realalloct , \bpricet \right\rangle\right]$, the expectation is exactly zero. The key observation is that after the first $t - 1$ rounds in the learning phase are finished, the price vector $\bpricet$ and the expected consumption $\cons(\bpricet)$ are both fixed, and therefore the expectation is zero because $\realalloct$ is an unbiased estimator of $\cons(\bpricet)$. When the factor $(1 - \pbe(\Lambda))$ is multiplied, note that the price vectors $\bprice^{(t+1)} , \cdots, \bprice^{(T)}$ all depend on the realization of $\realalloct$. Therefore, $\pbe(\Lambda)$ is correlated to the inner product $\left \langle  \cons(\bpricet) - \realalloct , \bpricet \right\rangle$, making it impossible to simply treat the expectation $\E_{\Lambda}\left[(1 - \pbe(\Lambda)) \cdot \sum_{t \in [T]} \left \langle  \cons(\bpricet) - \realalloct , \bpricet \right\rangle\right]$  as $0$.

To bypass this correlation issue, we crucially rely on the observation that after the price vector $\bpricet$ is fixed, the realization $\realalloct$ in the learning phase should be close to the expectation $\cons(\bpricet)$. This is simply because $\realalloct = \sum_{i \in [n]}\realallocit$, while each $\realallocit$ is an unbiased estimator of $\cons_i(\bpricet)$. To apply this idea, fix $\Lambda$, we first lower-bound $(1 - \pbe(\Lambda)) \cdot \eqref{eq:third}$ as 
\begin{align}
    (1 - \pbe(\Lambda)) \cdot \sum_{t \in [T]} \left \langle  \cons(\bpricet) - \realalloct , \bpricet \right\rangle ~&=~ -(1 - \pbe(\Lambda)) \cdot \sum_{t \in [T]} \left \langle  \realalloct - \cons(\bpricet) , \bpricet \right\rangle \notag\\
    ~&\geq~ -(1 - \pbe(\Lambda)) \cdot \sum_{t \in [T]} \left \langle  (\realalloct - \cons(\bpricet))^+ , \bpricet \right\rangle \notag \\
    ~&\geq~ -\sum_{t \in [T]} \left \langle  (\realalloct - \cons(\bpricet))^+ , \bpricet \right\rangle, \label{eq:third-1}
\end{align}
where notation $r^+$ represents $\max\{r, 0\}$ when $r$ is a single value, and we extend this notation to a vector, where notation $\boldsymbol{r}^+$ represents the vector obtained by taking the maximum of each coordinate and $0$, i.e., $\boldsymbol{r}^+ = (r_1, \cdots, r_m)^+ = (r_1^+, \cdots, r^+_m)$.

The benefit of introducing \eqref{eq:third-1} is to use the non-negativity of $(\realalloct -\cons(\bpricet) )^+$ to eliminate the term $\pbe(\Lambda)$. Then, the following lemma suggests that the expectation of \eqref{eq:third-1} is not too small:

\begin{restatable}{Lemma}{plussmall}
    \label{lma:plussmall}
    For \Cref{alg:exp}, we have 
    \[
    \E_{\Lambda}\left[\sum_{t \in [T]} \left \langle  (\realalloct - \cons(\bpricet))^+ , \bpricet \right\rangle \right] ~\leq~ \E_{\Lambda}\left[\sum_{t \in [T]} \left \langle  \epsilon \cdot \cons(\bpricet) + \epsilon \cdot \B , \bpricet \right\rangle \right].
    \]
\end{restatable}

The proof of \Cref{lma:plussmall} relies on the observation that $\realalloct$ should be close to its expectation $\cons(\bpricet)$, because the realization of $\realalloct$ consists of the realization of $n$ independent requests, and therefore $(\realalloct - \cons(\bpricet))^+$ can't be too large. We defer the proof of \Cref{lma:plussmall} to Section~\ref{subsec:proof-lma-plussmall}.

Now we take the expectation over $\Lambda$ for \eqref{eq:third-1} and apply \Cref{lma:plussmall}, which gives


\begin{align}
    \E_{\Lambda}\left[(1 - \pbe(\Lambda)) \cdot\sum_{t \in [T]} \left \langle  \cons(\bpricet) - \realalloct , \bpricet \right\rangle \right] ~&\geq~ -\E_{\Lambda}\left[\sum_{t \in [T]} \left \langle  \epsilon \cdot \cons(\bpricet) + \epsilon \cdot \B , \bpricet \right\rangle \right]. \label{eq:third-2}
\end{align}

The expectation $\E_{\Lambda}\left[\sum_{t \in [T]} \left \langle  \cons(\bpricet) , \bpricet \right\rangle \right]$ can be further upper-bounded by $\E_{\Lambda}[\algt(\Lambda)]$, because $\left \langle  \cons(\bpricet) , \bpricet \right\rangle$ represents the expected revenue collected by presenting $\bpricet$ to every request $i \in [n]$, which is further upper-bounded by the welfare collected by presenting $\bpricet$ to every request $i \in [n]$. 
Therefore, it remains to upper-bound $\E_{\Lambda}\left[\sum_{t \in [T]} \left \langle  \B , \bpricet \right\rangle \right]$. 
Note that rearranging \eqref{eq:furtheruse} implies that for every fixed $\Lambda$, we have
\begin{align*}
    \sum_{t \in [T]} \left \langle  \B , \bpricet \right\rangle ~&\leq~ \frac{1}{1- \epsilon} \cdot \sum_{t \in [T]} \left \langle  \realalloct , \bpricet \right\rangle + \frac{2nm \initprice}{\epsilon \cdot (1 - \epsilon)} + \frac{4\epsilon}{1 - \epsilon} \cdot \sum_{t \in [T]} \left \langle  \realalloct , \bpricet \right\rangle \\
    ~&\leq~ 6\sum_{t \in [T]} \left \langle  \realalloct , \bpricet \right\rangle + \frac{4nm \initprice}{\epsilon},
\end{align*}
where the second inequality holds when $\epsilon \leq 0.5$. Then, taking expectation over $\Lambda$ on both sides of the above inequality and applying to \eqref{eq:third-2} gives
\begin{align*}
    \E_{\Lambda}\left[(1 - \pbe(\Lambda)) \cdot\sum_{t \in [T]} \left \langle  \cons(\bpricet) - \realalloct , \bpricet \right\rangle \right] ~&\geq~ -7\epsilon \cdot \E_{\Lambda}[\algt(\Lambda)] - 4nm\initprice,
\end{align*}
where we also use 
\[
\E_{\Lambda}\left[\sum_{t \in [T]} \left \langle  \realalloct , \bpricet \right\rangle \right] = \E_{\Lambda}\left[\sum_{t \in [T]} \left \langle  \cons(\bpricet) , \bpricet \right\rangle \right] \leq \E_{\Lambda}[\algt(\Lambda)] 
\]
in the above inequality. Finally, applying \Cref{lma:algtsmall} gives
\begin{align}
     \E_{\Lambda}\left[(1 - \pbe(\Lambda)) \cdot\sum_{t \in [T]} \left \langle  \cons(\bpricet) - \realalloct , \bpricet \right\rangle \right] ~&\geq~ -14\epsilon \cdot T \cdot \opt - 4nm\initprice. \label{eq:third-final}
\end{align}

\vspace{0.5em} \noindent \textbf{Putting everything together.} Summing \eqref{eq:first-final}, \eqref{eq:second-final}, and \eqref{eq:third-final}, we have

\begin{align*}
    \E_\Lambda \left[\alg(\Lambda)\right] ~\geq~ \E_{\Lambda}\left[(1 - \pbe(\Lambda)) \cdot \algt(\Lambda)\right] ~&\geq~ \E_{\Lambda}\left[(1 - \pbe(\Lambda)) \cdot (\eqref{eq:first} + \eqref{eq:second} + \eqref{eq:third})\right] \\
    ~&\geq~ (1 - 26 \epsilon) \cdot T \cdot \opt - \frac{2nm \initprice}{\epsilon} - 4nm \initprice \\
    ~&\geq~ (1 - 32\epsilon) \cdot  T \cdot \opt,
\end{align*}
where the last inequality uses the definition of $\initprice=\esopt \cdot \frac{\epsilon^2}{m}$, the fact that $T \geq n$, and the assumption that $\opt \geq \widehat \opt$.

\bibliographystyle{alpha}
\bibliography{ref.bib}

\appendix
\section{Basic Probability Inequalities}

\begin{Theorem}[Generalized Freedman's Inequality, single-dimension case of Corollary 1.(b) in \cite{howard2020time}]
\label{thm:freedman} Let $\{Y_t\}_{t \in [T]}$ be a single-dimensional martingale whose increment satisfies $|\Delta Y_t| = |Y_t - Y_{t-1}| \leq 1$ for all $t$. Let $V_t = \sum_{i=1}^t \E[(Y_i - Y_{i-1})^2|\hist_{i-1}]$ be a random variable that represents the conditional variance of the martingale process. Then, for any $\eta, \zeta > 0$, we have
\begin{align*}
    \pr \left[\exists t \in [T]: |Y_t| \geq \eta + \frac{\eta/\zeta - \log(1 + \eta/\zeta)}{\log (1 + \eta/\zeta)} \cdot (V_t - \zeta)\right] ~\leq~ 2\exp \left(-\frac{\eta^2}{\zeta + \eta/3}\right).
\end{align*}
\end{Theorem}

\begin{Theorem}[Bernstein's Inequality for Bounded Variables]
\label{Bernstein-bounded}
Let $X_1, \ldots, X_N$ be independent mean-zero random variables such that $|X_i| \leq M$ for all $i$ and $\sigma^2 := \sum_{i \in [N]} \E[X^2_i]$ be the variance of $\sum_{i \in [N]} X_i$. Then, for any $\varepsilon \geq 0$,  we have
\begin{align*}
    \pr \left[\sum_{i = 1}^N X_i \geq \varepsilon \right] ~&\leq~ \exp\left(- \frac{\varepsilon^2/2}{\sigma^2 + M\varepsilon/3}\right) \quad \text{and} \\
    \pr \left[\sum_{i = 1}^N X_i \leq -\varepsilon \right] ~&\leq~ \exp\left(- \frac{\varepsilon^2/2}{\sigma^2 + M\varepsilon/3}\right) \enspace .
\end{align*}
\end{Theorem}

\section{Omitted proofs in \Cref{sec:sample}}
\label{sec:sample-omitted}

\subsection{Proof of \Cref{lma:budget-to-value}}
\label{subsec:lma-budget-to-value}

\begin{proof}
    Given the online resource allocation instance with generic request distributions $\bD = (\D_1, \cdots, \D_n)$ and budget $\B$,  we fix $\bprice$ to be the price vector that satisfies
    \[
    \cons(\bprice) ~\in~ [(1 - 3\epsilon) \cdot \B, (1 - \epsilon) \cdot \B].
    \]
    Note that $\bprice$ is a feasible solution of \ref{program:dual}$(\bD, \B)$ with the utility variables
    \[
    u_{i, \gamma_i} ~:=~ \max_{\theta \in \Theta_i} v_i(\theta) - \langle \bprice, \alloc_i(\theta) \rangle,
    \]
    where the feasibility naturally follows from the definitions of $ u_{i, \gamma_i}$ and the fact that the null action $\phi\in\Theta_i$, which guarantees $ u_{i, \gamma_i} \geq 0$. Since $u_{i, \gamma_i}$ also represents the utility when request $\gamma_i$ faces the price vector $\bprice$, we have
    \begin{align*}
        \wel(\bprice) ~&=~ \langle \bprice, \cons(\bprice) \rangle + \sum_{i \in [n]} \E_{\gamma_i \sim \D_i}[u_{i, \gamma_i}] \\
        ~&\geq~ (1 - 3\epsilon) \cdot \langle \bprice, \B \rangle + (1 - 3\epsilon) \cdot \sum_{i \in [n]} \E_{\gamma_i \sim \D_i}[u_{i, \gamma_i}] \\
        ~&\geq~ (1 - 3\epsilon)\cdot \opt,
    \end{align*}
    where the first line follows the fact that the total welfare can be decomposed as the total revenue plus the total utilities collected from each request, the second line uses the fact that $\cons(\bprice) \geq (1 - 3\epsilon) \cdot \B$ and the non-negativity of $u_{i, \gamma_i}$, and the last line follows the weak duality and the fact that $(\bprice, \{u_{i, \gamma_i}\})$ is a feasible solution of \ref{program:dual}$(\bD, \B)$.

    It remains to show that a static pricing mechanism with $\bprice$ being the pricing vector gets an expected welfare of at least $(1 - 4\epsilon) \cdot \opt$. For $i \in [n], j \in [m]$, define random variable 
    \[
    \rho_{i, j} ~:=~ a_{i, j} \left(\theta^*(\bprice; \gamma_i)\right) ~:~ \gamma_i = (v_i, \alloc_i, \Theta_i) \sim \D_i
    \]
    to be the random consumption of resource $j$ by request $i$ when facing price vector $\bprice$. We will show that for every $j \in [m]$, we have
    \begin{align}
        \pr \left[\sum_{i \in [n]} \rho_{i, j} > (1 - 0.5\epsilon) \cdot B_j\right] ~\leq~ \frac{\epsilon}{m}. \label{eq:bounded-consumption}
    \end{align}
    Proving \eqref{eq:bounded-consumption} is sufficient, as when \eqref{eq:bounded-consumption} holds, the expected welfare achieved by the static pricing mechanism can be lower bounded by
    \begin{align*}
        &\sum_{i \in [n]} \pr \left[\text{every resource } j\text{ has at least } 1 \text{ unit before }i \text{ arrives} \right] \cdot \E_{\gamma_i} \left[ v_i\big(\theta^*(\bprice; \gamma_i)\big)\right] \\
        ~\geq~& \sum_{i \in [n]} \left(1 - \sum_{j \in [m]} \pr \left[\sum_{i \in [n]} \rho_{i, j} > B_j - 1\right]\right) \cdot \E_{\gamma_i} \left[ v_i\big(\theta^*(\bprice; \gamma_i)\big)\right] \\
        ~\geq~& \sum_{i \in [n]} \left(1 - \sum_{j \in [m]} \pr \left[\sum_{i \in [n]} \rho_{i, j} > (1 - 0.5\epsilon) \cdot B_j\right]\right) \cdot \E_{\gamma_i} \left[ v_i\big(\theta^*(\bprice; \gamma_i)\big)\right] \\
        ~\geq~& (1 - \epsilon) \cdot \sum_{i \in [n]} \E_{\gamma_i} \left[ v_i\big(\theta^*(\bprice; \gamma_i)\big)\right] ~=~ (1 - \epsilon) \cdot \wel(\bprice) ~\geq~ (1 - 4\epsilon) \cdot \opt,
    \end{align*}
    where the second line follows from the union bound, the third line holds when $B_j \geq 2\epsilon^{-2}$, the last line applies \eqref{eq:bounded-consumption} and the fact that $\wel(\bprice) \geq (1 - 3\epsilon) \cdot \opt$.

    It remains to prove \eqref{eq:bounded-consumption}. Note that $\{\rho_{i,j} - \E[\rho_{i,j}]\}$ is a set of mean-zero random variables bounded in $[-1, 1]$, and by the definition of $\rho_{i, j}$, we have
    \[
    \sum_{i \in [n]} \E[\rho_{i, j}] ~=~ c_j(\bprice) ~\leq~ (1 - \epsilon) \cdot B_j,
    \]
    where the inequality uses the assumption that $\cons(\bprice) \leq (1 - \epsilon) \cdot \B$. Then, we have
\[
\var \left(\sum_{i \in [n]} \rho_{i, j} - \E[\rho_{i, j}]\right) ~=~ \sum_{i \in [n]} \var \left(\rho_{i, j} \right) ~\leq~ \sum_{i \in [n]} \E[\rho_{i, j}^2] ~\leq~ c_j(\bprice) ~\leq~ B_j
\]
By Bernstein's Inequality (\Cref{Bernstein-bounded}), we have
\begin{align*}
    \pr \left[\sum_{i \in [n]} \rho_{i, j} > (1 - 0.5\epsilon) \cdot B_j \right] ~&\leq~ \pr \left[\sum_{i \in [n]} \rho_{i, j} - c_j(\bprice) >  0.5 \epsilon \cdot B_j \right] \\
    ~&\leq~ \exp \left(- \frac{0.25\epsilon^2 \cdot B^2_j/2}{ B_j + 0.5\epsilon \cdot B_j/3}\right) ~\leq~ \frac{\epsilon}{m},
\end{align*}
where the last inequality holds when $\epsilon \leq 0.5$ and $B_j \geq 20 \log (m /\epsilon) \cdot \epsilon^{-2}$.
\end{proof}

\subsection{Proof of \Cref{lma:uniform-convergence}}
\label{subsec:lma-uniform-convergence}

\begin{proof}
    To prove \Cref{lma:uniform-convergence}, it's sufficient to show that for every $j \in [m]$, we have
    \begin{align}
        \pr\left[\sup_{\bprice} \left| \frac{\hat c'_j(\bprice)}{2B_j} - \frac{c'_j(\bprice)}{2B_j} \right| \geq \frac{\epsilon}{4}\right] ~\leq~ \frac{\delta}{m}. \label{eq:uniform-convergence}
    \end{align}
    Then, applying the union bound over all $j \in [m]$ proves \Cref{lma:uniform-convergence}.

    We now prove \eqref{eq:uniform-convergence}. Fix $j$. Let $\F$ be a class of functions, where each function $f_{\bprice}(\bgamma) \in \F$ is parameterized by a price vector $\bprice$, and maps $\bgamma = (\gamma_1, \cdots, \gamma_n)$ to $[0, 1]$, such that 
    \[
    f_{\bprice}(\bgamma) ~:=~ \frac{1}{2B_j} \cdot \min \left \{2B_j, \sum_{i \in [n]} a_{i,j}\big(\theta^*(\bprice; \gamma_i)\big)\right \}
    \]

    Let distribution $\bD$ be the original distribution of the given online resource allocation in \Cref{lma:uniform-convergence} and $\widehat \bD$ be the empirical distribution constructed via $N = C\cdot \frac{m}{\epsilon^2}\log^2\left(\frac{mn|\Theta|_{\max}}{\epsilon}\right)$ samples. For $f \in \F$, we define 
    \begin{align*}
        R(f_{\bprice}) ~:=~ \E_{\bgamma \sim \bD} [f_{\bprice}(\bgamma)] \qquad \text{and} \qquad \widehat R(f_{\bprice}) ~:=~ \E_{\bgamma \sim \widehat \bD} [f_{\bprice}(\bgamma)].
    \end{align*}

    Let $D$ be the Pseudo Dimension upper bound of function class $\F$. We will later show that
    \begin{align}
        D \leq O(m \log (m \cdot n \cdot |\Theta|_{\max})). \label{eq:pseudo}
    \end{align}
    Then, the standard uniform convergence bound via Pseudo Dimension (e.g., see \cite{Pollard-12}) gives that with probability at least $1 - \epsilon/m$, we have
    \begin{align*}
        \sup_{f_{\bprice} \in  \F} \big |R(f_{\bprice}) - \widehat R(f_{\bprice}) \big | ~\leq~ \sqrt{\frac{2D \log (2eN) + 2\log (4m/\epsilon)}{N}} ~\leq~ \frac{\epsilon}{4},
    \end{align*}
    where the last inequality holds when 
    \[
    N \geq C \cdot \frac{m}{\epsilon^2} \cdot \log (mn |\Theta|_{\max}/\epsilon) \cdot \log (m \cdot \epsilon^{-2} \cdot \log (mn |\Theta|_{\max}/\epsilon)) ~\geq~ C \cdot \frac{m}{\epsilon^2} \cdot \log^2 (mn |\Theta|_{\max}/\epsilon)
    \]
    for a sufficiently large $C$.

    It remains to show \eqref{eq:pseudo}. Let $\bgamma^{(1)}, \cdots, \bgamma^{(D)}$ be $D$ request profiles and let $(y^{(1)}, \cdots, y^{(D)}) \in [0, 1]^D$ be a set of thresholds that are shattered, that is, for every configuration $(t^{(1)}, \cdots, t^{(D)})$ such that each $t^{(d)} \in \{0, 1\}$, there exists a price vector $\bprice \in \R^m_{\geq 0}$, such that 
    \[
    \one[f_{\bprice}(\bgamma^{(d)}) > y^{(d)}] = t^{(d)}
    \]
    for every $d \in [D]$.

    Now consider two price vectors $\bprice, \bprice'$, such that there exists at least a $d \in [D]$ that satisfies $\one[f_{\bprice}(\bgamma^{(d)}) > y^{(d)}] \neq \one[f_{\bprice'}(\bgamma^{(d)}) > y^{(d)}]$. Then, there must exist at least a tuple $(d, i, \theta, \theta')$, such that $d \in [D], i \in [n], \theta, \theta' \in \Theta^{(d)}_i$, such that action $\theta$ is preferred under price vector $\bprice$, while $\theta'$ is preferred under $\bprice'$, i.e., we have
    \begin{align*}
        v^{(d)}_i(\theta) - \langle \alloc^{(d)}_i(\theta), \bprice \rangle ~&\geq~ v^{(d)}_i(\theta') - \langle \alloc^{(d)}_i(\theta'), \bprice \rangle, \quad \text{while} \\
        v^{(d)}_i(\theta') - \langle \alloc^{(d)}_i(\theta'), \bprice' \rangle ~&\geq~ v^{(d)}_i(\theta) - \langle \alloc^{(d)}_i(\theta), \bprice' \rangle,
    \end{align*}
    where one of the two inequalities is strict, as we apply a predefined deterministic tie-breaking rule that only depends on the request $\gamma^{(d)}_i$. Therefore, we have
    \begin{align*}
        v^{(d)}_i(\theta) - v^{(d)}_i(\theta') + \langle \alloc^{(d)}_i(\theta') -\alloc^{(d)}_i(\theta), \bprice \rangle ~&\geq~ 0, \quad \text{while} \\
        v^{(d)}_i(\theta) - v^{(d)}_i(\theta') + \langle \alloc^{(d)}_i(\theta') -\alloc^{(d)}_i(\theta), \bprice' \rangle ~&\leq~ 0,
    \end{align*}
    where again one of the two inequalities is strict. Therefore, the price vectors $\bprice$ and $\bprice'$ lie on different sides of the hyperplane defined by $v^{(d)}_i(\theta) - v^{(d)}_i(\theta') + \langle \alloc^{(d)}_i(\theta') -\alloc^{(d)}_i(\theta), \widetilde \bprice \rangle = 0$. 
    
    Since the tuples $(d, i, \theta, \theta')$ define at most $D \cdot n \cdot |\Theta|^2_{\max}$ hyperplanes, the space $\R^m_{\geq 0}$ is partitioned by these hyperplanes into no more than $m \cdot \left(D \cdot n \cdot |\Theta|^2_{\max}\right)^m$ regions (see \cite{Buck-43}). Then, if request profiles $\bgamma^{(1)}, \cdots, \bgamma^{(D)}$ and thresholds $(y^{(1)}, \cdots, y^{(D)}) \in [0, 1]^D$ are shattered, there must be
    \[
    2^D ~\leq~ m \cdot \left(D \cdot n \cdot |\Theta|^2_{\max}\right)^m,
    \]
    which implies $D \leq O(m \log (m \cdot n \cdot |\Theta|_{\max}))$.
\end{proof}

\section{Omitted Proofs from \Cref{sec:query-complexity}}

\subsection{Proof of \Cref{lma:pbesmall}}
\label{subsec:proof-lma-pbesmall}
\pbesmall*

We decompose the proof of \Cref{lma:pbesmall} into the following three claims:

\begin{Claim}
    \label{clm:realallocsmall}
    For \Cref{alg:exp}, we have 
    \[
    \pr_{\Lambda}\left[\exists j \in [m]: \sum_{t \in [T]} \realallocbase_j > (1 - 3\epsilon/4) \cdot T \cdot B_j\right] ~\leq~ \frac{\epsilon}{n},
    \]
    where we recall $\realallocbase_j$ represents the number of units of resource $j$ we allocate in the $t$-th round.
\end{Claim}

\begin{Claim}
    \label{clm:real-exp-diff-small}
    For \Cref{alg:exp}, we have
    \[
    \pr_{\Lambda}\left[\exists j \in [m]: \sum_{t \in [T]} \realallocbase_j + \epsilon/4 \cdot T \cdot B_j \leq \sum_{t \in [T]} \consbase_j(\bpricet)\right] ~\leq~ \epsilon,
    \]
    where we recall function $\consbase_j(\bprice)$ represents the expected unit of resource $j$ we allocate to $n$ requests, if every request responds to price vector $\bprice$.
\end{Claim}

\begin{Claim}
    \label{clm:exp-small-pbe-small}
    For any realization of $\Lambda$ in \Cref{alg:exp}, if the condition $ \sum_{t \in [T]} \cons(\bpricet) \leq (1 - \epsilon/2) \cdot T \cdot \B$ is satisfied, there must be $\pbe(\Lambda) \leq \epsilon$.
\end{Claim}

Combining \Cref{clm:realallocsmall}, \Cref{clm:real-exp-diff-small}, \Cref{clm:exp-small-pbe-small} proves \Cref{lma:pbesmall}:

\begin{proof}[Proof of \Cref{lma:pbesmall}]
    For simplicity of notation, define $\bd$ be the event that condition $ \sum_{t \in [T]} \cons(\bpricet) \leq (1 - \epsilon/2) \cdot T \cdot \B$ is satisfied for a realization $\Lambda$. Then \Cref{clm:exp-small-pbe-small} guarantees that $\pr_{\Lambda}[\bd] \leq \epsilon$. Note that
    \begin{align*}
        \E_\Lambda \left[\pbe(\Lambda)\right] ~=~& \E_\Lambda \left[\pbe(\Lambda) \Big| \bd\right] \cdot \pr_\Lambda \left[\bd \right] + \E_\Lambda \left[\pbe(\Lambda) \Big| \lnot \bd\right] \cdot \pr_\Lambda \left[\lnot \bd \right] \\
        ~\leq~& \epsilon \cdot 1 + 1 \cdot \pr_\Lambda \left[\lnot \bd \right].
    \end{align*}
    Therefore, it remains to upper bound the probability of event $\lnot \bd$. Note that when event $\bd$ does not happen, there must be at least one $j \in [m]$ satisfying either $\sum_{t \in [T]} \realallocbase_j > (1 - 3\epsilon/4) \cdot T \cdot B_j$ or $\sum_{t \in [T]} \realallocbase_j + \epsilon/4 \cdot T \cdot B_j \leq \sum_{t \in [T]} \consbase_j(\bpricet)$. Therefore, by the union bound, \Cref{clm:realallocsmall}, and \Cref{clm:real-exp-diff-small}, there must be $\pr_\Lambda[\lnot \bd] \leq \epsilon + \epsilon/n \leq 2\epsilon$, which further upper-bound $\E_\Lambda \left[\pbe(\Lambda)\right]$ by $3\epsilon$.
\end{proof}

Now, we finish the proofs of \Cref{clm:realallocsmall}, \Cref{clm:real-exp-diff-small}, and \Cref{clm:exp-small-pbe-small}.

\subsubsection{Proof of \Cref{clm:realallocsmall}}

It's sufficient to show that for every $j \in [m]$, we have
\begin{align}
    \pr_{\Lambda}\left[\sum_{t \in [T]} \realallocbase_j > (1 - 3\epsilon/4) \cdot T \cdot B_j\right] ~\leq~ \frac{\epsilon}{n\cdot m}, \label{eq:realalloc-sum-bounded}
\end{align}

as applying the union bound over all $j \in [m]$ proves \Cref{clm:realallocsmall}.

Fix the realization of requests $\{\gamma^{(t)}_i\}$. Suppose the inequality
\[
\sum_{t \in [T]} \realallocbase_j > (1 - 3\epsilon/4) \cdot T \cdot B_j
\]
holds. Let $\hat t < T$ be the index that satisfies
\[
\sum_{t = \hat t}^T \realallocbase_j \in [0.1 \epsilon \cdot T \cdot B_j, 0.2 \epsilon \cdot T \cdot B_j].
\]
Such index $\hat t$ must exist, as each $\realallocbase_j$, which represents the unit of consumed resource in round $t$, must fall between $[0, n]$, and there must be
\[
0.1 \cdot \epsilon \cdot T \cdot B_j ~\geq~ n
\]
when $T \geq n$ and $B_j \geq \frac{10}{\epsilon^2} \geq 10/\epsilon$ holds. Then, for every $t' \in [\hat t, T]$, the price $\price^{(t')}_j$ for resource $j$ must satisfy
\begin{align*}
    \price^{(t')}_j ~&=~ \initprice \cdot \exp \left(\frac{\epsilon}{n} \cdot \Big( \sum_{t = 1}^{t' - 1} \realallocbase_j - (1 - \epsilon) \cdot (t' - 1) \cdot B_j \Big)\right) \\
    ~&\geq~ \initprice \cdot \exp \left(\frac{\epsilon}{n} \cdot \Big( \sum_{t \in [T]} \realallocbase_j - (1 - \epsilon) \cdot T \cdot B_j - 0.2 \epsilon \cdot T \cdot B_j \Big)\right) \\
    ~&\geq~ \initprice \cdot \exp \left(\frac{\epsilon}{n} \cdot  0.05\epsilon \cdot T \cdot B_j  \right) \\
    ~&\geq~ \initprice \cdot \exp\big(10 \log(n \cdot m \cdot \beta/\epsilon) \big) ~\geq~ \frac{10nm \cdot \opt}{\epsilon^2},
\end{align*}
where the last line holds when $T \geq n$ and $B_j \geq 200 \log (nm \beta/\epsilon) \cdot \epsilon^{-2}$, and we use the assumption that $\opt \leq \beta \cdot \widehat \opt$ in the last inequality. 

With the above lower bound for $\price^{(t')}_j$, note that the total value collected from round $\hat t$ to round $T$ is at least the payment for buying resource $j$. Then, when the inequality $\sum_{t \in [T]} \realallocbase_j > (1 - 3\epsilon/4) \cdot T \cdot B_j$ holds,  the total value collected in the learning phase of \Cref{alg:exp} is at least
\[
\frac{10nm \cdot \opt}{\epsilon^2} \cdot 0.1 \epsilon \cdot T \cdot B_j ~\geq~ \frac{nm}{\epsilon} \cdot T \cdot \opt.
\]
Note that the above inequality implies 
\[
    \pr_{\Lambda}\left[\sum_{t \in [T]} \realallocbase_j > (1 - 3\epsilon/4) \cdot T \cdot B_j\right] ~\leq~ \frac{\epsilon}{n\cdot m},
\]
as if the bound does not hold, the expected value collected in the learning phase of \Cref{alg:exp} is at least $T \cdot \opt$, which is in contrast to the fact that no algorithm can collect more than $T \cdot \opt$ value in the learning phase.

\subsubsection{Proof of \Cref{clm:real-exp-diff-small}}

   It's sufficient to show that for every $j \in [m]$, we have
    \begin{align}
    \label{eq:real-exp-diff-small}
        \pr_{\Lambda}\left[\sum_{t \in [T]} \realallocbase_j + \epsilon/4 \cdot T \cdot B_j \leq \sum_{t \in [T]} \consbase_j(\bpricet)\right] ~\leq~ \frac{\epsilon}{m}.
    \end{align}
    Then, taking union bound over $j \in [m]$ proves \Cref{clm:real-exp-diff-small}.

    Fix $j$ and the realization in the learning phase. To prove \eqref{eq:real-exp-diff-small}, note that when the event $\sum_{t \in [T]} \realallocbase_j + \epsilon/4 \cdot T \cdot B_j \leq \sum_{t \in [T]} \consbase_j(\bpricet)$ happens, there must be either
\[
    \sum_{t \in [T]} \realallocbase_j > (1 - 3\epsilon/4) \cdot T \cdot B_j \quad \text{or} \quad \Big| \sum_{t \in [T]} \realallocbase_j - \sum_{t \in [T]} \consbase_j(\bpricet) \Big| ~\geq~ \frac{0.1n\epsilon \cdot B_j + 0.1\epsilon \cdot \sum_{t \in [T]} \realallocbase_j}{1 - 0.1\epsilon},
\]
    as if neither of the two events hold, there must be
\begin{align*}
     \Big| \sum_{t \in [T]} \realallocbase_j - \sum_{t \in [T]} \consbase_j(\bpricet) \Big| ~&\leq~ 0.11 T\epsilon \cdot B_j + 0.11 \epsilon \cdot \sum_{t \in [T]} \realallocbase_j \\
    ~&\leq~ 0.22 \epsilon \cdot T \cdot B_j ~\leq~ \epsilon/4 \cdot T \cdot B_j,
\end{align*}
where the first inequality holds when $\epsilon \leq 0.5$ and $T \geq n$, and the second inequality uses the assumption that $\sum_{t \in [T]} \realallocbase_j \leq (1 - 3\epsilon/4) \cdot T \cdot B_j \leq T \cdot B_j$.

    Therefore, we have
\begin{align*}
    \pr_{\Lambda} &\left[\sum_{t \in [T]} \realallocbase_j + \epsilon/4 \cdot T \cdot B_j \leq \sum_{t \in [T]} \consbase_j(\bpricet)\right] ~\leq~ \pr_{\Lambda}\left[\sum_{t \in [T]} \realallocbase_j > (1 - 3\epsilon/4) \cdot T \cdot B_j\right]\\
    & + \pr_{\Lambda}\left[\Big| \sum_{t \in [T]} \realallocbase_j - \sum_{t \in [T]} \consbase_j(\bpricet) \Big| ~\geq~ \frac{0.1n\epsilon \cdot B_j + 0.1\epsilon \cdot \sum_{t \in [T]} \realallocbase_j}{1 - 0.1\epsilon}\right] \\
    \leq& ~\frac{\epsilon}{n  \cdot m} + \pr_{\Lambda}\left[\Big| \sum_{t \in [T]} \realallocbase_j - \sum_{t \in [T]} \consbase_j(\bpricet) \Big| ~\geq~ \frac{0.1n\epsilon \cdot B_j + 0.1\epsilon \cdot \sum_{t \in [T]} \realallocbase_j}{1 - 0.1\epsilon}\right],
\end{align*}
where the second inequality follows from \eqref{eq:realalloc-sum-bounded}, and it remains to show
\begin{align}
    \pr_{\Lambda}\left[\Big| \sum_{t \in [T]} \realallocbase_j - \sum_{t \in [T]} \consbase_j(\bpricet) \Big| ~\geq~ \frac{0.1n\epsilon \cdot B_j + 0.1\epsilon \cdot \sum_{t \in [T]} \realallocbase_j}{1 - 0.1\epsilon}\right] ~\leq~ \frac{\epsilon}{m}. \label{eq:ineq-via-martingale}
\end{align}

    To prove \eqref{eq:ineq-via-martingale}, we model the consumption of resource $j$ as a martingale process. Let $\{Y_t\}_{t \in [T]}$ be a martingale, such that 
    \[
    \Delta Y_t ~=~ Y_t - Y_{t-1} ~:=~ \frac{\realallocbase_j - \consbase_j(\bpricet)}{n}
    \]
    is a mean-zero random variable bounded between $[-1, 1]$, since when the requests till round $t - 1$ are realized, $\consbase_j(\bpricet)$ represents the expectation of $\realallocbase_j$, and we allocate at most $n$ unit of resource $j$ in round $t$. Let $\hist_{t}$ denote the history of requests till round $t$, and let
    \[
    V_t ~:=~ \sum_{t' = 1}^t \var(\Delta Y_{t'} ~|~ \hist_{t' - 1})
    \]
    be the conditional variance of the martingale process. Note that when the requests till round $t - 1$ are realized, we have
    \begin{align*}
        \var(\Delta Y_{t} ~|~ \hist_{t-1}) = \var \left( \frac{\realallocbase_j}{n} ~|~ \hist_{t-1} \right) \leq \E \left[\frac{\realallocbase_j}{n} \cdot \frac{\realallocbase_j}{n} ~|~ \hist_{t-1}  \right] \leq \E \left[\frac{\realallocbase_j}{n}  ~|~ \hist_{t-1} \right] = \frac{\consbase_j(\bpricet)}{n},
    \end{align*}
    where the first equality uses the fact that $\consbase_j(\bpricet)$ is deterministic when the history $\hist_{t-1}$ till round $t - 1$ is fixed and does not affect the conditional variance; the second inequality uses the fact that $\var(X)\leq \E(X^2)$ for any random variable $X$; and the second inequality uses the fact that $\frac{\realallocbase_j}{n} \in [0, 1]$. Therefore,
    \begin{align}
        V_T ~=~ \sum_{t \in [T]} \var(\Delta Y_{t} ~|~ \hist_{t-1}) ~\leq~ \frac{1}{n} \cdot \sum_{t \in [T]} \consbase_j(\bpricet) ~\leq~ \frac{1}{n} \cdot \sum_{t \in [T]} \realallocbase_j + |Y_T|. \label{eq:martingale-var-bound}
    \end{align}

    Now, we apply \Cref{thm:freedman} with $\eta = 0.1\epsilon \cdot B_j$ and $\zeta =  B_j$ to prove \Cref{eq:ineq-via-martingale}.  We have
    \begin{align*}
        &\pr_\Lambda \left[\Big| \sum_{t \in [T]} \realallocbase_j - \sum_{t \in [T]} \consbase_j(\bpricet) \Big| ~\geq~ \frac{0.1n\epsilon \cdot B_j + 0.1\epsilon \cdot \sum_{t \in [T]} \realallocbase_j}{1 - 0.1\epsilon} \right] \\
        ~=~& \pr_\Lambda \left[|Y_T| ~\geq~ \eta + 0.1\epsilon \cdot \left( \frac{1}{n} \cdot \sum_{t \in [T]} \realallocbase_j + |Y_T|\right) \right] \\
        ~\leq~& \pr_\Lambda \left[|Y_T| ~\geq~ \eta + 0.1\epsilon \cdot V_T \right] \\
        ~\leq~& \pr_\Lambda \left[|Y_T| ~\geq~ \eta + \frac{\eta/\zeta - \log(1 + \eta/\zeta)}{\log(1 + \eta/\zeta)} \cdot (V_T - \zeta) \right] \\
        ~\leq~& 2\exp \left( -\frac{\eta^2}{\zeta + \eta/3} \right) \\
        ~=~& 2\exp\left(-\frac{0.01\epsilon^2B_j^2}{B_j+0.1\epsilon B_j/3}\right) \\
        ~\leq~& 2\exp \left(-\epsilon^2 B_j / 1000 \right) ~\leq~ \epsilon/m,
    \end{align*}
    where the third line uses \eqref{eq:martingale-var-bound}; the fourth line uses the fact that $\eta/\zeta = 0.1\epsilon \leq 1$, and inequality $\frac{x - \log (1 + x)}{\log(1 + x)} \leq  x$ holds when $x \in (0, 1]$; the fifth line applies \Cref{thm:freedman}; the sixth line applies the value of $\eta$ and $\zeta$; and the last inequality holds when $B_j \geq 10000 \log (m/\epsilon) \cdot \epsilon^{-2}$.

\subsubsection{Proof of \Cref{clm:exp-small-pbe-small}}

It's sufficient to show that for each resource $j \in [m]$, the probability that the auction phase allocates more than $B_j - 1$ unit of resource $j$ is at most $\epsilon/m$. Then, applying union bound over $j \in [m]$ proves \Cref{clm:exp-small-pbe-small}. 

Fix $j$ and $\Lambda$. Let $X_{i,j} \in [0, 1]$ be the random variable representing the unit of resource $j$ consumed by request $i$ in the auction phase. Then, the desired inequality is
\[
\pr \left[ \sum_{i \in [n]} X_{i, j} ~\geq~ B_j - 1\right] ~\leq~ \frac{\epsilon}{m}.
\]

We prove the above inequality via concentration. 
Since request $i$ faces a price vector chosen from $\Lambda$ uniformly at random, we have
\[
\E[X_{i,j}] ~=~ \frac{1}{T} \cdot \sum_{t \in [T]} \consbase_{i,j} (\bpricet), \quad \text{and} \quad \var(X_{i, j}) = \E[X^2_{i, j}] ~\leq~ \E[X_{i, j}] ~=~ \frac{1}{T} \cdot \sum_{t \in [T]} \consbase_{i,j} (\bpricet).
\]
Then, since \Cref{clm:exp-small-pbe-small} assumes $ \sum_{t \in [T]} \cons(\bpricet) \leq (1 - \epsilon/2) \cdot T \cdot \B$, we have
\begin{align}
    \sum_{i \in [n]} \var(X_{i, j}) ~\leq~ \sum_{i \in [n]} \E[X_{i,j}] ~=~ \frac{1}{T} \cdot \sum_{t \in [T]} \consbase_j(\bpricet) ~\leq~ (1 - \epsilon/2) \cdot B_j \label{eq:consbase-sum-bounded} 
\end{align}

Now, we apply \Cref{Bernstein-bounded} to the summation of mean-zero random variables $X_{i, j} - \E[X_{i, j}]$, which gives
\begin{align*}
    \pr \left[ \sum_{i \in [n]} X_{i, j} ~\geq~ B_j - 1\right] ~&\leq~ \pr \left[ \sum_{i \in [n]} X_{i, j} ~\geq~  (1 - \epsilon/4) \cdot B_j\right] \\
    ~&=~ \pr \left[ \sum_{i \in [n]} \left(X_{i, j} - \E[X_{i,j}]\right) ~\geq~  (1 - \epsilon/4) \cdot B_j - \sum_{i \in [n]} \E[X_{i,j}]\right] \\
    ~&\leq~ \pr \left[ \sum_{i \in [n]} \left(X_{i, j} - \E[X_{i,j}]\right) ~\geq~ \epsilon/4 \cdot B_j \right] \\
    ~&\leq~ \exp \left(-\frac{\epsilon^2 \cdot B^2_j /32}{\sum_{i \in [n]} \var(X_{i, j}) + \epsilon \cdot B_j / 12} \right) \\
    ~&\leq~ \exp\left( -\epsilon^2 \cdot B_j/32\right) ~\leq~ \epsilon/m,
\end{align*}
where the third line and the last line follows \eqref{eq:consbase-sum-bounded}, the fourth line follows \Cref{Bernstein-bounded}, and the last inequality holds when $B_j \geq 32 \log(m/\epsilon) \cdot \epsilon^{-2}$.

\subsection{Proof of \Cref{lma:algtsmall}}
\label{subsec:proof-lma-algtsmall}

\algtsmall*

\begin{proof}

We first observe that conditioning on the learning phase till round $t - 1$, the expected total value collected in the $t$-th round of the learning phase is exactly $\sum_{i \in [n]} \wel_i(\bpricet)$. Taking the expectation over the learning phase till round $t - 1$ and summing this equality over $t \in [T]$, it implies the expected total value collected by the learning phase of \Cref{alg:exp} is exactly $\E_{\Lambda}[\algt(\Lambda)]$. Therefore, to prove \Cref{lma:algtsmall}, it's equivalent to show that the total value collected in the learning phase of \Cref{alg:exp} is at most $2T \cdot \opt$.

To bound the total value collected in the learning phase, the main idea is to show that the total budget consumption in the learning phase is bounded by $T \cdot \B$ with high probability. To be specific, we consider the following two stages algorithm as an alternate learning phase:
\begin{itemize}
    \item Stage 1: Run the learning phase algorithm of \Cref{alg:exp}, until some resource $j$ is allocated at least $T \cdot B_j - n$ units at the end of round $\tau < T$, or reaching $T$ rounds.
    \item Stage 2: For the remaining rounds $\tau + 1, \cdots, T$, satisfy the highest value action of each request.
\end{itemize}

Clearly the total value collected in the learning phase of \Cref{alg:exp} is upper-bounded by the total value collected in the above alternate algorithm. It remains to show that the alternate algorithm collects at most $2 T \cdot \opt$ value. For Stage 1 of the algorithm, the total budget consumption is at most $T \cdot \B$, and therefore the total value collected in Stage 1 can't exceed $T \cdot \opt$. For Stage 2, the total collected value is bounded by
\begin{align*}
    \sum_{t \in [T]} \pr[t \text{ is in Stage 2}] \cdot \sum_{i \in [n]} \E_{\gamma_i}\left[\max_{\theta \in \Theta_i} v_i(\theta)\right] ~&\leq~ \sum_{t \in [T]} \pr[t \text{ is in Stage 2}] \cdot n \cdot \opt \\
    ~&\leq~ T \cdot n \cdot \opt \cdot \pr[T \text{ is in Stage 2}],
\end{align*}
where the first inequality uses the fact that the expected maximum value of a request can't be more than $\opt$, and the second inequality uses the fact that the probability that a round $t$ is in Stage 2 is non-increasing with respect to $t$.

It remains to bound $\pr[T \text{ in Stage 2}]$. Note that when the event ``$T$ is in Stage 2'' happens, there must be a resource $j$, such that the Stage 1 of the algorithm, which is equivalent to the learning phase of \Cref{alg:exp}, allocates at least $T \cdot B_j - n > T \cdot B_j - 3\epsilon/4 \cdot T \cdot B_j$ unit of resource $j$ in $T$ rounds. Since \Cref{clm:realallocsmall} guarantees that the learning phase of \Cref{alg:exp} allocates more than $(1 - 3\epsilon/4) \cdot T \cdot B_j$ unit of resource $j$ for some $j \in [m]$ with probability at most $\epsilon/n$, there must be $\pr[T \text{ in Stage 2}] \leq \epsilon/n$. Therefore, the total expected value collected by the alternate algorithm is at most 
\[
T \cdot \opt + \frac{\epsilon}{n} \cdot T \cdot n \cdot \opt ~\leq~ 2T \cdot \opt. \qedhere
\]
\end{proof}

\subsection{Proof of \Cref{lma:plussmall}}
\label{subsec:proof-lma-plussmall}

\plussmall*

\begin{proof}
    Note that
    \begin{align*}
        \E_{\Lambda}\left[\sum_{t \in [T]} \left \langle  (\realalloct - \cons(\bpricet))^+ , \bpricet \right\rangle \right] ~=~ \sum_{t \in [T]} \E_{t-1} \left[ \left \langle  \E_{\realalloct}\left[(\realalloct - \cons(\bpricet))^+ \right] , \bpricet \right\rangle \right],
    \end{align*}
    where the notation $\E_{t-1}$ represents taking the expectation over only the first $t-1$ rounds of the learning phase, and the equality follows from the fact that price vector $\bpricet$ is fixed after the first $t-1$ rounds of the learning phase are realized. Similarly, we have
    \[
    \E_{\Lambda}\left[\sum_{t \in [T]} \left \langle  \epsilon \cdot \cons(\bpricet) + \epsilon \cdot \B , \bpricet \right\rangle \right] ~=~ \sum_{t \in [T]}  \E_{t-1}\left[\left \langle  \epsilon \cdot \cons(\bpricet) + \epsilon \cdot \B , \bpricet \right\rangle \right].
    \]
    Therefore, to prove \Cref{lma:plussmall}, it's sufficient to show that for any realization of the first $t - 1$ rounds of the learning phase, there must be 
    \begin{align*}
        \E_{\realalloct}\left[(\realalloct - \cons(\bpricet))^+ \right] ~\leq~ \epsilon \cdot \cons(\bpricet) + \epsilon \cdot \B ,
    \end{align*}
    which is further equivalent to show that for every $j \in [m]$, we have
    \begin{align}
        \E[(\realallocbase_j - \consbase_j(\bpricet))^+] \leq \epsilon \cdot \consbase_j(\bpricet) + \epsilon \cdot B_j.\label{eq:plussmall-1} 
    \end{align}

    Let random variable $Z_i = \realallocbase_{i,j} - \consbase_{i,j}(\bpricet)$, i.e., the difference between the real unit of resource $j$ consumed by request $i$ in round $t$ and its difference. Then, we have $Z_i \in [-1, 1]$, as each request consumes at most one unit of each resource, and $\E[Z_i] = 0$. Note that
    \begin{align*}
        \E^2\left[(\realallocbase_j - \consbase_j(\bpricet))^+\right] ~=~ \E^2\left[\Big(\sum_{i \in [n]} Z_i\Big)^+\right] ~&\leq~ \E\left[\Big(\sum_{i \in [n]} Z_i\Big)^2\right] \\
        ~&=~  \var  \left(\sum_{i \in [n]} Z_i\right) \\
        ~&\leq~ \sum_{i \in [n]} \var(Z_i) \\
        ~&\leq~ \sum_{i \in [n]} \var(\realallocbase_{i,j}) ~\leq~ \sum_{i \in [n]} \consbase_{i,j} (\bpricet) ~=~ \consbase_j(\bpricet),
    \end{align*}
    where the first line uses the standard Cauchy-Schwarz inequality for random variables, together with the fact that $(r^+)^2 \leq r^2$ for every value $r$, the second line uses the fact that each $Z_i$ is a mean-zero random variable, the third line uses the fact that each $Z_i$ is independent, and the last line uses the fact that $\realallocbase_{i,j} \in [0,1]$ to upper-bound its variance as its expectation. Therefore, we have
\[
\E\left[(\realallocbase_j - \consbase_j(\bpricet))^+\right] ~\leq~ \sqrt{\consbase_j(\bpricet)} ~\leq~ \epsilon \cdot \consbase_j(\bpricet) + \epsilon \cdot B_j,
\]
where the last inequality holds from the fact that we have $\sqrt{\consbase_j(\bpricet)} \leq \sqrt{B_j} \leq \epsilon \cdot B_j$ when $\consbase_j(\bpricet) \leq B_j$, while $\sqrt{\consbase_j(\bpricet)} \leq \consbase_j(\bpricet) / \sqrt{B_j} \leq \epsilon \cdot \consbase_j(\bpricet)$ when $ \consbase_j(\bpricet) \geq B_j$.
\end{proof}

\end{document}